\documentclass[aps, pra, nofootinbib, onecolumn, 11pt, tightenlines, notitlepage, superscriptaddress, longbibliography]{revtex4-1}

\usepackage{graphicx, amsmath, amssymb, amsthm, xcolor, bm, bbm}

\theoremstyle{plain}
\newtheorem{thm}{Theorem}
\newtheorem{lem}[thm]{Lemma}
\newtheorem{prop}[thm]{Proposition}
\newtheorem{cor}[thm]{Corollary}
\theoremstyle{definition}
\newtheorem{defn}[thm]{Definition}

\theoremstyle{remark}

\renewenvironment{proof}[1][Proof]{\noindent\textbf{#1.} }{\ $\Box$}

\definecolor{nblue}{rgb}{0.2,0.2,0.7}
\definecolor{ngreen}{rgb}{0.1,0.5,0.1}
\definecolor{nred}{rgb}{0.8,0.2,0.2}
\definecolor{nblack}{rgb}{0,0,0}

\newcommand{\A}{\mathbb{A}}

\newcommand{\D}{\mathcal{D}}
\newcommand{\G}{\mathcal{G}}
\newcommand{\E}{\mathcal{E}}
\newcommand{\M}{\mathcal{M}}

\newcommand{\R}{\mathbb{R}}

\newcommand{\C}{\mathbb{C}}

\newcommand{\GL}{\mathsf{GL}}
\newcommand{\U}{\mathsf{U}}
\newcommand{\unit}{\mathbbm{1}}

\newcommand{\mc}[1]{\mathcal{#1}}
\newcommand{\mbb}[1]{\mathbb{#1}}

\newcommand{\density}[1]{|#1\rangle\langle #1|}
\newcommand{\inner}[2]{\langle #1,#2\rangle}
\newcommand{\ket}[1]{|#1\rangle}
\newcommand{\bra}[1]{\langle #1|}
\newcommand{\tr}{\mathrm{Tr}}

\newcommand{\nt}[1]{^{\otimes #1}}
\newcommand{\norm}[1]{\left\lVert #1\right\rVert}
\newcommand{\av}[1]{\left\lvert #1\right\rvert}
\newcommand{\ct}{^{\dagger}}

\newcommand{\e}{\textrm{e}}

\newcommand{\pa}[1]{\left(#1\right)}
\newcommand{\bpa}[1]{\bigl(#1\bigr)}

\newcommand{\br}[1]{\left\{#1\right\}}

\newcommand{\sq}[1]{\left[#1\right]}
\newcommand{\bsq}[1]{\bigl[#1\bigr]}

\newcommand{\bgsq}[1]{\biggl[#1\biggr]}

\usepackage{hyperref}
\definecolor{darkblue}{RGB}{0,0,127} 
\definecolor{darkgreen}{RGB}{0,150,0}
\hypersetup{breaklinks, colorlinks, linkcolor=darkblue, citecolor=darkgreen, filecolor=red, urlcolor=blue}
\hypersetup{pdftitle={Randomized Benchmarking with Confidence}, pdfauthor={Joel J. Wallman and Steven T. Flammia}}

\begin{document}
\title{Randomized Benchmarking with Confidence}

\author{Joel J.\ \surname{Wallman}}
\affiliation{Institute for Quantum Computing, University of Waterloo, Waterloo, Canada}
\affiliation{Centre for Engineered Quantum Systems, School of Physics, The University of Sydney, Sydney, NSW 2006, Australia}
\author{Steven T.\ \surname{Flammia}}
\affiliation{Centre for Engineered Quantum Systems, School of Physics, The University of Sydney, Sydney, NSW 2006, Australia}

\date{\today}

\begin{abstract}
Randomized benchmarking is a promising tool for characterizing the noise in experimental implementations of quantum systems. In this paper, we prove that the estimates produced by randomized benchmarking (both standard and interleaved) for arbitrary Markovian noise sources are remarkably precise by showing that the variance due to sampling random gate sequences is small. We discuss how to choose experimental parameters, in particular the number and lengths of random sequences, in order to characterize average gate errors with rigorous confidence bounds. We also show that randomized benchmarking can be used to reliably characterize time-dependent Markovian noise (e.g., when noise is due to a magnetic field with fluctuating strength). Moreover, we identify a necessary property for time-dependent noise that is violated by some sources of non-Markovian noise, which provides a test for non-Markovianity.
\end{abstract}

\maketitle

\section{Introduction}

One of the key obstacles to realizing large-scale quantum computation is the need for error correction and fault tolerance~\cite{Gottesman2009}, which require the coherent implementation of unitary operations to high precision. Characterizing the accuracy of an experimental implementation of a unitary operation is therefore an important prerequisite for constructing a large-scale quantum computer.

It is possible to completely characterize an experimental implementation of a unitary using full quantum process tomography~\cite{Chuang1997, Poyatos1997}. However, this approach has several major deficiencies when applied to large quantum systems. Firstly, it is provably exponential in the number of qubits of the system for any procedure that can identify general noise sources and hence it cannot be performed practically for even intermediate numbers of qubits, despite improvements such as compressed sensing~\cite{Flammia2012, Gross2010a}. Secondly, it is sensitive to state preparation and measurement (SPAM) errors, which create a noise floor below which an accurate process estimation becomes impossible~\cite{Merkel2012}. Finally, it does not capture any notion of systematic, time-dependent errors that can arise from applying many unitaries in sequence.

One can avoid the exponential scaling by accepting a partial characterization of an experimental implementation. A partial characterization of, for example, the average error rate and/or the worst-case error rate compared to a perfect implementation of a target unitary is typically enough to determine whether an experimental implementation of a unitary is sufficient for achieving fault-tolerance in a specific scheme for fault-tolerant quantum computation. Such partial characterizations can be obtained efficiently (in the number of quantum systems) using either randomized benchmarking~\cite{Emerson2005, Knill2008, Magesan2011, Magesan2012a, Gaebler2012, Magesan2012} or direct fidelity estimation~\cite{DaSilva2011, Flammia2011}. 

While direct fidelity estimation gives an unconditional and assumption-free estimate of the average gate fidelity, it is prone to state preparation and measurement (SPAM) errors, which leads to conflation of noise sources. Thus, a key advantage of randomized benchmarking is that it is not sensitive to SPAM errors. Unfortunately, however, current proposals for randomized benchmarking assume that the noise is time-independent, although time-dependence can be partially characterized by a deviation from the expected fidelity decay curve~\cite{Magesan2011, Magesan2012a}. Furthermore, experimental implementations of randomized benchmarking typically use on the order of 100 random sequences of Clifford gates, which is three orders of magnitude smaller than the number of sequences suggested by the rigorous bounds in Ref.~\cite{Magesan2012a} to obtain an accuracy comparable to the claimed experimental accuracies~\cite{Magesan2012, Brown2011}. Numerical investigations of a variety of noise models have shown that between 10--100 random sequences for each length are sufficient to provide a tight estimate of the average gate fidelity~\cite{Epstein2014}. Ideally, one would like to combine the advantages of both randomized benchmarking and direct fidelity estimation to achieve a method that is insensitive to SPAM, requires few measurements, is nearly assumption-free (i.e., does not assume a specific noise model), and comes with rigorous guarantees on the errors involved. 

In this paper, we provide a new analysis of randomized benchmarking which brings it closer in line with this ideal. We first show that the standard protocol can be modified to provide a means of estimating the time-dependent average gate fidelity (which characterizes the average error rate), provided that the gate-dependent fluctuations at each time step are sufficiently small. Under the assumption that the noise is Markovian (that is, that the noise can be written as a sequence of noisy channels acting on the system of interest), all the time-dependent parameters that are estimated by our procedure are upper-bounded by 1, so if some of the parameters are observed to be greater than 1, the experimental noise must be non-Markovian.

We then provide a rigorous justification for taking a small number of random sequences at each length that is on the same order as used in practice by obtaining bounds on the variance due to sampling gate sequences. Our work complements the approach of Ref.~\cite{Epstein2014}, where it was shown that the width of the confidence interval for the parameters extracted from randomized benchmarking is on the order of the square root of the variance. Our work therefore proves that this confidence interval is generally very narrow, that is, the parameters extracted from randomized benchmarking are determined with high precision. 

Numerically, we observe that our bounds (at least for qubits) are saturated and so cannot be improved without further assumptions on the noise (e.g., that the noise is diagonal in the Pauli basis). Therefore any experiments using fewer random sequences than justified by our analysis (unless there is solid evidence that the noise has a specific structure) will potentially underestimate the error due to sampling random sequences.

As a particular example, our results provide a rigorous proof that for single-qubit noise with an average error rate of $10^{-4}$, the error for randomized benchmarking with 100 random sequences of 100 random gates will be less than $0.9\%$ with $99\%$ confidence. If we use the parameters estimated in the experiment of Ref.~\cite{Brown2011}, with 100 random sequences of length 987 at an average error rate of $2\times 10^{-5}$, we find the error is less than $.8\%$ with $99\%$ confidence.

We emphasize that our results are solely in terms of the number of random gate \emph{sequences}, and a given sequence must still be repeated many times to gather statistics about expectation values of an observable. This is of course an unavoidable consequence of quantum mechanics. However, these statistical fluctuations in the estimates of expectation values can be analyzed separately with standard statistical tools for binomial distributions or with the recent Bayesian methods introduced in~\cite{Granade2014} and combined seamlessly with our results.

In order to give a rigorous statement of results, we will first review the randomized benchmarking protocol.

\section{The randomized benchmarking protocol}\label{sec:protocol}

The goal of randomized benchmarking is to efficiently but partially characterize the average noise in an experimental implementation of a group $\G=\{g_1,\ldots,g_{\av{\G}}\}\subset \U(d)$ of operations acting on a $d$-dimensional quantum system. In order to characterize the average noise in an implementation of $\G$ using randomized benchmarking, we require $\G$ to be a unitary 2-design (e.g., the Clifford group on $n$ qubits for $d=2^n$), meaning that sampling over $\G$ reproduces the second moments of the Haar measure~\cite{Dankert2009,Gross2007a}. To accomplish this, the following protocol is implemented.
\begin{itemize}
	\item Choose a random sequence $s=s_1\ldots s_m\in\mbb{N}_{\av{G}}^m$ of $m$ integers chosen uniformly at random from $\mbb{N}_{\av{G}}=\br{1,\ldots,\av{\G}}$.
	\item Prepare a $d$-dimensional system in some state $\rho$ (usually taken to be the pure state $\ket{0}$). 
	\item At each time step $t=0,\ldots,m$, apply $g_t$ where $g_t = g_{s_t}$ and $g_0 := \prod_{t=1}^{m} g_t^{-1}$. Alternatively, to perform interleaved randomized benchmarking for the gate $g_{\rm int}\in\G$, apply $g_{t,{\rm int}}$ where $g_{t,{\rm int}} = g_{\rm int} g_t$ for $t\neq 0$ and, as before, $g_{0,{\rm int}} = \prod_{t=1}^{m} g_{t,{\rm int}}^{-1}$. (In general, each gate must be compiled into a sequence of elementary gates as well.)
	\item Perform a POVM $\br{E,\unit - E}$ for some $E$ (usually taken to be $|0\rangle\!\langle0|$) and repeat with the sequence $s$ sufficiently many times to obtain an estimate of the probability $F_{m,s} = p(E|s,\rho)$ to a suitable precision.
\end{itemize}

We can regard the probability $F_{m,s}$ as a realization of a random variable $F_m$. We will denote the variance of the distribution $\{F_{m,s}:s\in\mbb{N}_{\av{G}}\}$ for a fixed $m$ by $\sigma_m^2$. Averaging $F_{m,s}$ over a number of random sequences will give an estimate $\hat{F}_m$ of $\bar{F}_m$, the average of $F_{m,s}$ over all sequences $s$ of fixed length $m$ (that is, $\bar{F}_m$ is the expectation of the random variable $F_m$). The accuracy of this estimate will be a function of the number of random sequences and $\sigma_m^2$. 

Obtaining estimates $\hat{F}_m$ for multiple $m$ and fitting to the model
\begin{align}
	\bar{F}_m = A + B f^m
\end{align}
will give an estimate of $f$ provided that the noise does not depend too strongly on the target gate~\cite{Magesan2012a}, where~\cite{Nielsen2002}
\begin{align}
	f = \frac{d\mc{F}_{\rm avg}(\mc{E}) - 1}{d-1}
\end{align}
and
\begin{align}
	\mc{F}_{\rm avg}(\mc{E}) &= \int \mathrm{d}\psi \tr\bsq{\psi\mc{E}(\psi)}
\end{align}
is the average gate fidelity of a noise channel $\mc{E}$ with respect to the identity channel and $\mathrm{d}\psi$ is the uniform Haar measure over all pure states. The average gate fidelity of $\mc{E}$ gives the average probability that preparing a state $\psi$, applying $\mc{E}$ and then measuring $\{\psi,\unit-\psi\}$ will give the outcome $\psi$, averaged over all pure states $\psi$.

For standard randomized benchmarking, $\mc{E}$ is the error channel per operation, averaged over all operations in $\mc{G}$. For interleaved benchmarking, $\mc{E}$ is the error channel on a composite channel, namely, the interleaved channel composed with an element of $\G$, averaged over all $\G$. We note in passing that separating the error in the interleaved channel from the error in the composite channel is one of the key difficulties in obtaining meaningful results from interleaved benchmarking~\cite{Kimmel2013}, though we do not address this issue here.

\section{Statement of Results and Paper Outline}\label{sec:results}

The first principal contribution of this paper is to show that the number of random sequences that need to be averaged is comparable to the number actually used in contemporary experiments (compared to previous best estimates, which require 3 orders of magnitude more random sequences than currently used). The second principal contribution is to show that randomized benchmarking can be used to characterize time-dependent fluctuations in the noise strength.

In more detail, and in order of appearance, we show the following.
\begin{itemize}
	\item We use the results derived later in the paper to obtain explicit confidence intervals for the estimates $\hat{F}_m$ when $mr\ll 1$, where $r = 1-\mc{F}_{\rm avg}(\mc{E})$ is the average gate infidelity (Sec.~\ref{sec:confidence}).
	\item Again, using results derived later, we show that a more thorough analysis of randomized benchmarking data can be used to characterize time-dependent Markovian noise, and consequently as a sufficient condition for the presence of non-Markovian noise in a system (Sec.~\ref{sec:time_dependent}).
	\item We review representation theory and the Liouville representation of quantum channels and prove some elementary results (Sec.~\ref{sec:preliminaries}). We give an explicit proof of bounds on the diamond norm (which characterizes the worst-case error rate) in terms of the average gate fidelity (which characterizes the average error rate). These give slight improvements over previously stated (but unproven) bounds (Sec.~\ref{sec:fidelities}).
	\item We derive an expression for the mean of the 
	randomized benchmarking distribution with time-dependent noise (Sec.~\ref{sec:mean}).
	\item  We show that the variance for randomized benchmarking $d$-level systems with average gate infidelity and sequences of length $m$ satisfies
	\begin{align}
		\sigma_m^2 \leq 4d(d+1)mr + O(m^2 r^2 d^4)\,.
	\end{align}
	Furthermore, we provide an argument that suggests that this bound can be improved to
	\begin{align}
		\sigma_m^2 \leq mr + O(m^2 r^2 d^4)\,.
	\end{align}
	
	\item For qubits, we improve the upper bound to
	\begin{align}
		\sigma_m^2 \leq m^2 r^2 + \frac{7 mr^2}{4} + 6\delta mr +  O(m^2 r^3) + O(\delta m^2 r^2)\,,
	\end{align}
	where $\delta$ quantifies the deviation from preparations and measurements in a Pauli eigenstate. We use this improved bound to derive confidence intervals that rigorously justify the use of a small number of random sequences for qubits in the regime $mr\ll 1$.
	
	\item For the special case of single-qubit noise that is diagonal in the Pauli basis, we further improve the upper bound to
	\begin{align}
		\sigma_m^2 \leq \frac{11 mr^2}{4} + O(m^2 r^3)	\,,
	\end{align}
	which is independent of preparations and measurements.
	
	\item We show that the variance for unital (but nonunitary) channels decays exponentially to zero asymptotically, 
	while the variance for nonunital noise converges exponentially to a 
	positive constant proportional to the degree of nonunitality (as suitably quantified).
	
	\item We prove that our results are robust under gate-dependent noise, which is one of the key assumptions under which randomized benchmarking produces a meaningful result. Furthermore, since our results apply to interleaved randomized benchmarking, gate dependence can be experimentally tested and used to bound the contribution from gate-dependent terms.
\end{itemize}

\section{Analyzing data from randomized benchmarking with finite sampling}\label{sec:analysis}

In this section, we summarize the implications of our results for analyzing the data obtained from randomized benchmarking experiments. In particular, we derive confidence intervals for the estimates $\hat{F}_m$ of $\bar{F}_m$ and show how randomized benchmarking can be used to characterize time-dependent noise.

\subsection{Confidence interval for randomized benchmarking}\label{sec:confidence}

For a fixed sequence length $m$, randomized benchmarking provides an estimate $\hat{F}_m$ of $\bar{F}_m$, which is exact in the limit when all random sequences are sampled. We will only consider the variance $\sigma_m^2$ due to sampling a finite number $K_m$ of random sequences of length $m$, and we ignore the random fluctuations resulting from the use of a finite number of measurements to estimate a probability.

In Ref.~\cite{Magesan2012a}, the variance-independent form of Hoeffding's inequality was used to estimate the number of sequences $K_m$ required to obtain a given level of accuracy. The estimate in Ref.~\cite{Magesan2012a} erroneously restricted the range of the random variable in Hoeffding's inequality. That is, they assumed that all the probabilities $F_{m,s}$ lay in a strict subset of $[0,1]$. This assumption, while valid for depolarizing noise, is not valid in general. A simple counterexample is where the noise is a single-qubit preparation channel into the $\density{0}$ state and $\rho = E = \density{0}$. Then any sequence of $m$ gates ending in an identity gate or a $z$-axis rotation has $F_{m,s} = 1$, while any sequence ending in an $X$ gate gives $F_{m,s}=0$. Correcting for this (which does not change any of the conclusions of Ref.~\cite{Magesan2012a}), the variance-independent form of Hoeffding's inequality requires $10^5$ samples to ensure that the estimate $\bar{F}_m$ is within $5\times 10^{-3}$ of the true mean $\bar{F}_m$ with $99\%$ probability. However, many experimental implementations of randomized benchmarking only use 30--100 sequences for each value of $m$~\cite{Brown2011, Gambetta2012, Magesan2012}. 

One of the principal contributions of this paper is to provide a theoretical justification for choosing a relatively small number of sequences by showing that the variance is small for the short sequences that are of practical relevance. For the special case of qubits, we show that even for small $m$ (e.g.\ $m\approx 100$) the variance is at most $4\times 10^{-4}$ for currently achievable gate infidelities $r\approx 10^{-4}$, which is comparable to the numerical estimates presented in Fig.~\ref{fig:numerical}. Utilizing this very small variance gives substantial improvements over the previous rigorous bounds obtained in Refs.~\cite{Magesan2012a, Kimmel2013}. However, our bound on the variance (which is numerically almost optimal for qubits) implies that $K_m$ should scale \emph{quadratically} with $m$ to make the variance is independent of $m$.

Our upper bound $\sigma_m^2 \leq m^2 r^2 + \tfrac{7}{4} m r^2 + O(m^2 r^3)$ (for qubits, neglecting the negligible $\delta r$ terms) can be used together with a stronger version of Hoeffding's inequality~\cite{Hoeffding1963} to obtain a rigorous confidence interval comparable to the standard errors of the mean reported in current experiments~\cite{Brown2011}. The stronger version of Hoeffding's inequality implies that
\begin{align}
	\Pr\biggl(\av{\hat{F}_m - \bar{F}_m} > \epsilon\biggr) \leq 2 \Bigl[H(\epsilon,\sigma_m^2)\Bigr]^{K}\,,
\end{align}
where $K$ is the number of randomly sampled sequences of length $m$ and 
\begin{align}
	H(\epsilon,v) = \Bigl(\frac{1}{1-\epsilon}\Bigr)^{\frac{1-\epsilon}{v+1}} \Bigl(\frac{v}{v+\epsilon}\Bigr)^{\frac{v+\epsilon}{v+1}}\,.
\end{align}
Consequently, sampling
\begin{align}
	K = -\frac{\log\bigl(2/\delta\bigr)}{\log\bigl(H(\epsilon,\sigma_m^2)\bigr)}
\end{align}
random sequences is sufficient to obtain an absolute precision of $\epsilon$ with probability $1-\delta$. Since $r$ is determined by the fitting procedure, which in turn depends on the uncertainties, this procedure would be applied recursively with an initial upper bound on $r$. Similarly, for qudits, a straightforward generalization of the above argument can be used, but with $\sigma_m^2 \leq 4d(d+1)mr + O(d^4 m^2 r^2)$. (There are various inefficiencies in this estimate which mean that it does not reduce to the same answer as above for $d=2$; see Theorem~\ref{thm:qudit_bound} for more details.)

To get a feel for the sort of estimates that this bound provides, consider the following parameters for a single-qubit benchmarking experiment: $m=100, r=10^{-4}, \epsilon = 1\%, \delta = 1\%$, and use our upper bound of $\sigma^2_m = m^2 r^2 + \tfrac{7}{4} m r^2$ (ignoring the higher-order terms). Then our bound shows that $K = 145$ random sequences suffices. This is an improvement by \textit{orders of magnitude} over the previous best rigorously justifiable upper bound of $10^5$ using the variance-independent Hoeffding inequality~\cite{Magesan2012a}. 

Importantly, however, we note that the quadratic scaling with $m$ in the regime $m r \ll 1$ seems to be necessary (see Fig.~\ref{fig:numerical}). Even in the optimal case of noise that is diagonal in the Pauli basis, $K_m$ would still need to scale linearly with $m$ to make the variance independent of $m$ (where having the variance depend on $m$ would generally cause less weight to be assigned to larger $m$ when fitting). The linear scaling can be understood intuitively as following from the fact that there are $m$ places for an error in a sequence of length $m$, and the errors could add up in the worst case. Therefore, a corollary of our result is that longer sequence lengths should be averaged over more random sequences in this regime.

Furthermore, we prove in Sec.~\ref{sec:asymptotic} that there are noise sources such that the variance due to sampling random sequences is constant (or decays on an arbitrarily long timescale). If such noise sources (including nonunital noise and any unitary noise, such as over- and under-rotations) are believed to be present, substantially more random sequences need to be sampled. As such, randomized benchmarking is most reliable in the regime $mr\ll 1$, although, since the next lowest order terms in our bound are $\delta m^2 r^2$ and $m^2 r^3$, the lowest order bounds on the variance should be approximately valid for $mr\approx 0.1$.

\begin{figure}[tb]
	\includegraphics{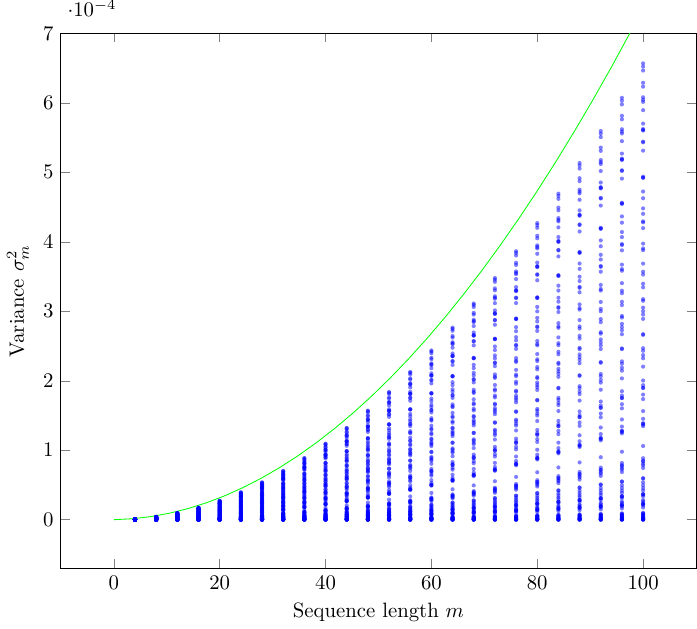}
	\caption{\label{fig:numerical} Plot of a random sampling of the exact variance $\sigma_m^2$ as a function of the sequence length $m$ for randomized benchmarking with 100 randomly generated, time-independent noisy qubit channels. The noise was sampled from the set of extremal qubit channels, characterized in Ref.~\cite{Ruskai2002}, with average gate infidelity $r \lesssim 2.69\times 10^{-4}$. Each channel was evolved for increasing sequence lengths to track the behavior of the variance as a function of $m$, which is why the data points track parabolic curves (furthermore, the spread in the parabolic curves is generated by the spread in the infidelity of the samples). The green curve is the upper bound $\sigma_m^2 = m^2 r^2 + \tfrac{7 m r^2}{4}$, where we neglect the corrections to the bound at order $O(m^2 r^3)$ and corrections due to measurement imprecision. Note that our bound is almost optimal. Our analytic results show that the variance $\sigma_m^2$ will \emph{increase} with $m$, at least until some threshold sequence length where the exponential decay for generic channels proven in Theorem~\ref{thm:asymptotic} begins to dominate.}
\end{figure}

\subsection{Characterizing time-dependent noise}\label{sec:time_dependent}

The original presentation of randomized benchmarking assumed that the noise was approximately time independent (i.e., independent of the time step at which the gate is applied), with any Markovian time-dependence being partially characterized by deviations from the time-independent case~\cite{Magesan2012a}. However, in many practical applications there may be a nonnegligible time dependence, which it would be desirable to characterize more fully.

We show that randomized benchmarking can also be used to characterize time-dependent noise, provided the gate-dependence is negligible (in the sense established in Theorem~\ref{thm:gate-perturb}) and that the time-dependent noise is identically distributed between different experiments. However, the number of random sequences of length $m$ will typically need to be increased relative to the number required for time-independent noise. In particular, we will show in Theorem~\ref{thm:fidelity_curve} that
\begin{align}\label{eq:fidelity_curve}
	\bar{F}_m =  A + B\prod_{t=1}^{m} f_t \,,
\end{align}
where $A$ and $B$ are constants that depend only upon the preparation and measurement procedures (and so account for SPAM) and the average gate fidelity at time $t$ is $F_t = f_t + (1-f_t)/d$, where $d$ is the dimensionality of the system being benchmarked. In the case of time-independent noise, $f_t$ is a constant and Eq.~\eqref{eq:fidelity_curve} reduces to the standard equation for the fidelity decay curve.

By performing randomized benchmarking for a set of sequence lengths $m_1$ and $m_2$, we can estimate $\hat{F}_{m_j}$ with associated uncertainties $\delta_j$. Combining these estimates with a procedure for obtaining an estimate $\hat{A}$ of $A$ with associated uncertainty $\delta_A$~\cite{Kimmel_in_prep}, we can estimate the ratio
\begin{align}\label{eq:estimate_product}
	\frac{\bar{F}_{m_2} - \bar{A}}{\bar{F}_{m_1} - \bar{A}} = \prod_{t=m_1 + 1}^{m_2}f_t
\end{align}
with uncertainty on the order of
\begin{align}
	\delta_{1,2,A} \approx \sqrt{\pa{\delta_1+\delta_A}^2 + \pa{\delta_2+\delta_A}^2}	\,.
\end{align}

Therefore we can estimate the average gate infidelity $r$ over the time interval $[m_1 + 1,m_2]$. From our rigorous analysis, we can infer that $\delta_1$ and $\delta_2$ will be small for small $m_j r$, while $\delta_A$ will be determined only by finite measurement statistics. Furthermore, when $m_2\approx m_1$, $\prod_{t=m_1 + 1}^{m_2}f_t\approx 1 - r(m_2-m_1)$ and so the above method gives a reliable method of characterizing the time-dependent gate fidelity.

We also note that if there are no temporal correlations in the noise, than all of the parameters $r_t$ (where $r_t$ is the average gate infidelity at time $r$) are lower-bounded by zero. Therefore any negative values (or average values) of $r_t$ are an indicator of temporal correlations in the noise, that is, of non-Markovian behavior.

\section{Mathematical preliminaries}\label{sec:preliminaries}

Randomized benchmarking involves composing random sequences of quantum channels that are 
sampled in a way which approximates a group average. For this reason, it is 
natural to consider both the representation theory of groups and the structure 
of quantum channels, especially the composition of channels. In this section 
we collect several mathematical results in this vein that we will need to 
prove our main results. We begin by considering group representation theory, 
and in particular prove a proposition showing how the tensor product of certain 
representations couple together. Most of this material is standard and can be found in any textbook on the subject, e.g.~\cite{Goodman2009}.

\subsection{Representation Theory and Some Useful Lemmas}\label{sec:rep_theory}

A \emph{representation} (rep) of a group $\G$ is a pair $(\phi,V)$, where $V$ 
is vector space known as the representation space (which we always take to be 
$\mbb{R}^d$ or $\mbb{C}^d$ for different values of $d$) and 
$\phi:\G\to \GL(V)$---where $\GL(V)$ is the general linear group over $V$---is
a homomorphism. A rep is \emph{faithful} if $\phi$ is injective and 
\emph{unitary} (resp. \emph{orthogonal}) if $\phi(g)$ is a unitary (resp. 
orthogonal) operator for all $g\in\G$. The dimension of a rep is the 
dimension of $V$. A \emph{subrepresentation} (subrep) is a pair $(\phi_W,W)$ 
such that $\phi(g)W\subseteq W$ for all $g\in\G$ and $\phi_W$ denotes the 
restriction of $\phi$ to the subspace $W$. We sometimes refer to a space, 
subspace, or homomorphism as being a rep or subrep, with the complementary 
ingredients understood from the context. 

A rep is called \emph{irreducible} or an \emph{irrep} if the only subreps are
$\emptyset$ and $V$. Since the reps we consider are unitary reps of compact
groups, if $W$ is a subrep of $V$ then the orthogonal complement $W^{\bot}$ is 
a subrep as well. Therefore any rep can be decomposed into a direct sum of
irreps, which may occur with some multiplicity. Any basis that decomposes a rep
into a direct sum of irreps is called a Schur basis.

The simplest rep is the trivial rep $(1,\C)$, which is also an irrep.
The trivial rep is defined for any group $\G$ and take any element of $\G$ to 
1. While the trivial rep deserves its name, it frequently appears as a subrep 
of tensor powers of other reps and so will appear throughout this paper.

The randomized benchmarking protocol is designed so that the sequence of 
operators applied to a system correspond to noise channels conjugated by 
uniformly random elements of a group $\G$. Given a rep $(\phi,V)$ of a group $\G$, a 
matrix $A\in\GL(V)$ and an element $g\in\G$, we define $A^g = 
\phi(g)A\phi(g^{-1})$. The uniform average of this action on $A$ is called the \emph{$\G$-twirl} of $A$, and is given by $A^\G = \av{\G}^{-1}\sum_{g\in\G} A^g$.

Note that, for notational convenience, the map $\phi$ is left implicit but will always be obvious given the dimensionality of the matrix being twirled. An important property of $A^\G$ is that it commutes with the action of $\G$ for \textit{any} rep $(\phi,V)$ (reducible or not) 
since $\phi$ is a homomorphism and $\G$ is a group. That is, $A^\G = (A^g)^\G = (A^\G)^g$ for all $g \in \G$.  

Expressions for the expected value $\bar{F}_m$ and variance $\sigma_m^2$ for the randomized benchmarking protocol for a fixed value of $m$ will be obtained using the following propositions.

\begin{prop}\label{prop:twirling}
	Let $(\phi,\C^d)$ be a nontrivial $d$-dimensional irreducible representation of a group $\G$ and $A\in\GL(\C^d),B\in\GL(\C^{d+1})$. Then
	\begin{itemize}
		\item $A^\G  = a\unit_d$;
		\item $B^\G = B_{11}\oplus b\unit_d$ [where the representation of $\G$ is $(1\oplus\phi,\C^{d+1})$]; and
		\item $\sum_{g\in\G}\phi(g)=0$,
	\end{itemize}
	where $a=\tr A/d$ and $b=(\tr B-B_{11})/d$.
\end{prop}

\begin{proof}
	All three statements follow directly from Schur's Lemma~\cite{Goodman2009}.
\end{proof}

\begin{prop}\label{lem:trivial_subrep}
	If $(\phi,V)$ is an irreducible representation of a finite group $\G$ with a real-valued character $\chi_\phi$, then the trivial representation is a subrepresentation of $(\phi,V)\nt{2}$ with multiplicity 1.
\end{prop}

\begin{proof}
	As the rep is irreducible, Schur's orthogonality relations~\cite{Goodman2009} give
	\begin{align}
		\av{\G} = \sum_{g\in\G} \chi_{\phi}(g)^*\chi_{\phi}(g) = \sum_{g\in\G} 
		\chi_{\phi}(g)^2 = \sum_{g\in\G} \chi_{\phi\nt{2}}(g)\chi_{1}(g)	\,,
	\end{align}
	where we have used $\chi_{\phi\nt{2}}(g) = \sq{\chi_{\phi}(g)}^2$ and that the 
	character for the trivial representation is $\chi_1(g) = 1$ for all $g\in\G$.
\end{proof}

\subsection{The Liouville Representation of Quantum Channels}\label{sec:channels}

A \textit{quantum channel} is a linear map $\E:\D_{d_1}\to\D_{d_2}$, where 
$\D_d$ is the set of $d$-dimensional density operators. 
Quantum channels can be represented in a variety of equivalent 
ways, with different representations naturally suited to particular 
applications.

In this paper, we will primarily use the Liouville representation 
because it is defined so that quantum channels compose under matrix 
multiplication. We occasionally also use the Choi representation in order to 
apply results from the literature, but we will introduce it only as required.

\subsubsection{States and measurements}

We begin by introducing the Liouville representation (also called the transfer matrix representation) of quantum states 
and measurements. States and measurement effects (i.e., elements of a positive-operator valued measure, or POVM) can be viewed as channels from 
$\E:\R\to\D_d$ and $\E:\D_d\to\R$ respectively, hence they can be treated on 
the same footing as any other quantum channel. However, we introduce them 
separately for pedagogical and notational clarity.

In the standard formulation of quantum mechanics in terms of density operators 
and POVMs, a quantum state $\rho\in\D_d$ 
is any Hermitian, positive semi-definite operator such that $\tr\rho=1$. 
In addition, we always have
$\tr\rho^2\in \sq{0,1}$. We can always choose a basis 
$\A=\br{A_0,A_1,\ldots,A_{d^2-1}}$ of orthonormal operators for $\GL(\C^d)$, 
where orthonormality is according to the Hilbert-Schmidt inner product 
$\inner{A}{B} = \tr\pa{A\ct B}$. We can expand any density operator relative 
to such a basis as $\rho = \sum_j \rho_j A_j$, where 
$\rho_j = \inner{A_j}{\rho}$. Throughout this paper we set 
$A_0 = \unit/\sqrt{d}$, which fixes $\rho_0 = \tr\rho/\sqrt{d}$, and makes all other $A_j$ for $j\not=0$ traceless. 

We can then identify a density operator $\rho$ with a corresponding column vector 
\begin{align}\label{eq:bloch_vector}
	|\rho) = \pa{\begin{array}{c}
		\rho_0 \\
		\vec{\rho}
	\end{array}}\in\C^{d^2}
\end{align} 
such that $\vec{\rho}_j = \rho_j$ for $j=1,\ldots,d^2-1$. Here we make the 
important distinction between the density operator itself, $\rho$, and the 
representation of $\rho$ in terms of the column vector $|\rho)$. Note that 
$|\rho)$ is just a generalized version of a Bloch vector (with a different 
normalization) for $d\ge2$. 

The conditions for $\rho$ to correspond to a density operator now translate 
into geometric conditions on $|\rho)$. In particular, we will use the fact that 
$\norm{|\rho)}_2^2 = \tr \rho^2 \in\sq{0,1}$, where $\norm{v}_2$ for 
$v\in\C^{d^2}$ is the standard isotropic Euclidean norm. 

Measurements in the standard formulation correspond to POVMs, that is, to sets of Hermitian, positive semidefinite operators $\br{E_j}$ such that $\sum_j E_j = \unit$. As with quantum states, we can expand an element $E$ of a POVM (an effect) as $E = \sum E_j A_j\ct$, where $E_j = \inner{E}{A_j}$. We then identify an effect $E$ with a row vector 
\begin{align*}
	(E| = \bpa{\begin{array}{cc}
		E_0 & \vec{E}
	\end{array}}\in {\C^*}^{d^2}	\,,
\end{align*}
which must satisfy similar conditions to $|\rho)$. 

In this formalism, the probability of observing an effect $E$ given that the quantum state $\rho$ was prepared is $p(E|\rho) = \tr E\rho = (E|\rho)$.

\subsubsection{Transformations}

For simplicity, we will only consider quantum channels that are either states, 
measurements or completely positive and trace-preserving (CPTP) maps 
$\E:\D_d\to\D_d$. We do not consider channels that reduce the trace or change the dimension because, while conceptually no more difficult, they require cumbersome additional notation and we do not use any such channels.

A quantum channel $\E$ maps a density operator $\rho$ to another density 
operator $\E(\rho)$. We want to determine the map $\E$ between the 
corresponding vectors $|\rho)$ and $(\E(\rho)|$. Since quantum channels are 
linear,
\begin{align*}
	\E(\rho) = \sum_j \E(A_j)\rho_j \,,
\end{align*}
which implies that
\begin{align*}
	\bigl\lvert\E(\rho)\bigr)_k = \sum_j \inner{A_k}{\E(A_j)}\bigl\lvert\rho_j\bigr) \,.
\end{align*}
That is, $\bigl\lvert\E(\rho)\bigr) = \E |\rho)$ where we abuse notation slightly and 
define $\E$ as the matrix such that $\E_{j,k} = \inner{A_k}{\E(A_j)}$. That 
is, we use $\E$ to denote both the abstract operator as well as its 
representation as a matrix acting on vectors $|\rho)$. In this 
representation, the identity channel is represented by $\unit_{d^2}$, the 
composition of two channels is given by matrix multiplication, and furthermore, 
the conjugate channel of a unitary channel $\E$ is given by $\E\ct$. 
In particular, these properties imply that the Liouville representation 
of the unitary channels is a faithful and unitary representation of $\U(d)$ (though technically, it is a projective representation since a global phase is lost). 
Given our choice of $\A$ (recall that we have fixed $A_0$) and the fact that we consider only trace-preserving channels, 
we will always write the matrix representation of a quantum channel as
\begin{align}
	\E = \pa{\begin{array}{cc} 1 & 0 \\ \alpha(\E) & \varphi(\E)\end{array}}	\,.
\end{align}
A channel $\E$ is unital (i.e., the identity is a fixed point of the channel) iff $\alpha(\E)=0$. Therefore we can regard $\norm{\alpha(\E)}$ as quantifying the nonunitality of $\E$.

The representation $\bigl(\varphi,\C^{d^2-1}\bigr)$ of $\U(d)$ is irreducible~\cite{Gross2007a}, 
which will play a crucial role in our analysis of randomized benchmarking 
since it allows us to use tools from representation theory such as Schur's 
Lemma. Note also that the representation $(\varphi,\C^{d^2-1})$ of any subgroup $\G\subseteq 
\U(d)$ that is a unitary 2-design is also irreducible by the same argument 
(which can be regarded as a defining property of a unitary 
2-design~\cite{Gross2007a}). Therefore we can also use tools from 
representation theory when considering channels twirled over a unitary 
2-design. This fact allows randomized benchmarking to be performed efficiently
because unitary 2-designs can be efficiently sampled while the full 
unitary group cannot~\cite{Emerson2005, Dankert2009}.

The representation $\varphi(g)$ of $\G$ will be one of the basic tools we use in this paper. As such, whenever $g$ appears in a matrix multiplication, it will refer to $\varphi(g)$.

Randomized benchmarking will allow for the estimation of
\begin{align}
	f(\E):=\tfrac{1}{d^2-1}\tr\varphi(\E)\,,
\end{align}
which corresponds to the average gate fidelity of $\E$ with the identity 
channel. We will sometimes omit the argument of $\alpha$, $\varphi$ and $f$, 
or indicate the argument via a subscript. However, in all cases the argument 
will be clear from the context.

\subsubsection{Properties of channels in the Liouville representation}

Since the Liouville representation associates a unique matrix to each channel, we 
can characterize properties of quantum channels by properties of the 
corresponding matrix. In particular, we will consider the spectral radius, 
\begin{align}
	\varrho(M) = \max_j \av{\eta_j(M)}\,,
\end{align}
and the spectral norm, $\norm{M}_{\infty} = \max \sigma_j(M)$, where $\br{\eta_j(M)}$ and $\br{\sigma_j(M)}$ are the eigenvalues and singular values of a matrix $M$ respectively. These norms satisfy 
$\varrho(M)\leq \norm{\M}_{\infty}$, which we will use to obtain bounds on valid quantum channels.

\begin{prop}
	Let $\E$ be a completely positive map. Then the adjoint channel $\E\ct$ is also completely positive.
\end{prop}

\begin{proof}
	Any map can be written as
	\begin{align}
		\E = \sum K_j\otimes L_j^T	\,,
	\end{align}
	where the superscript $T$ denotes the transpose and the $K_j$ and $L_j$ are Kraus operators for $\mc{E}$. We then have
	\begin{align}
		\E\ct = \sum K_j\ct \otimes L_j^*\,,
	\end{align}
	where the $*$ denotes complex conjugation.
	
	By Choi's theorem on completely positive maps, $\mc{E}$ is completely positive if and only if Kraus operators can be chosen so that $L_j = K_j\ct$. Therefore Kraus operators $K_j\ct$ and $L_j\ct$ for $\E\ct$ can be chosen so that $L_j\ct = (K_j\ct)\ct$.
\end{proof}

\begin{cor}\label{cor:adjoint}
	The adjoint channel of a unital, completely-positive and trace-preserving channel is also a unital, completely-positive and trace-preserving channel.
\end{cor}

\begin{prop}\label{prop:channel_norm}
	Any completely positive and trace-preserving channel $\E:\D_d\to \D_d$ satisfies the following relations: 
	\begin{alignat}{3}
		& \textrm{(i)} \quad& \det \E & \leq 1\,, \notag\\
		& \textrm{(ii)} \quad& \norm{\E}_{\infty} & \leq \sqrt{d}\,, \notag\\
		& \textrm{(iii)} \quad& \varrho(\E) & \leq 1\,, \notag\\
		& \textrm{(iv)} \quad& \norm{\alpha(\E)}_2 & \leq \sqrt{d-1}\,.
	\end{alignat}
	Inequality (i) is saturated if and only if $\E$ is unitary.
	
	Furthermore, if $\E$ is unital, then (ii) can be improved to $\norm{\E}_{\infty} = 1$.
\end{prop}

\begin{proof}
	(i): See Ref.~\cite[Thm 2]{Wolf2008}.
	
	(ii): See~\cite[Thm. II.1]{Perez-Garcia2006}, noting that $\|\E\|_\infty = \|\E\|_{2 \to 2}$. 
	
	(iii): See~\cite{Evans1978}.
	
	(iv): For any density operator $\rho$, we have $\tr\rho^2 \in \sq{0,1}$. In 
	particular consider $\E(\unit/d)$, which must be a density operator since 
	$\unit/d$ is a density operator and $\E$ is a quantum channel. Then
	\begin{align}
		\tr\,\E(\unit/d)^2 = \norm{\E B(\unit/d)}_2^2 = \frac{1 + \norm{\alpha(\E)}_2^2}{d} \le 1
	\end{align}
	gives the desired result.
	
	Now let $\E$ be a unital channel. Then $\E\ct$ and thus $\E\ct\E$ are also channels by Corollary~\ref{cor:adjoint}. Substituting (i) into the equality $\varrho(\E\ct \E) = \norm{\E}_{\infty}$ gives the improved bound.
\end{proof}

\subsection{Representing noisy channels}

An attempt to physically implement a quantum channel $\E$ will generally result in some other channel $\E'$, with the aim being, loosely speaking, to make $\E'$ as close to $\E$ as possible. We will now outline how noisy channels can be related to the intended channel in the linear representation.

Consider an attempt to implement a target unitary channel $\mc{U}$ that results in some (noisy) channel $\E$. Then since $\E$ is a real square matrix, it can be written as $\E = LQ$ where $L$ is a lower triangular matrix and $Q$ is an orthogonal matrix. Since $\mc{U}$ is an orthogonal matrix, we can always write $\E = \Lambda^{\rm post}(\mc{U}) \mc{U}$, where $\Lambda^{\rm post}(\mc{U}) = LQ\mc{U}^T$. Similarly, we can always write $\E = \mc{U}\Lambda^{\rm pre}(\mc{U})$.

While the difference between these expressions is trivial for any \emph{single} 
channel, it can cause confusion when comparing channels. Since $\Lambda^{\rm 
	pre}(\mc{U}) = \mc{U}^T\Lambda^{\rm post}(\mc{U}) \mc{U}$, the notion of ``the'' 
noise in an implementation of $\mc{U}$ depends on which expression is used. (We 
will fix a representation below to avoid this ambiguity.)

The convergence of randomized benchmarking depends crucially upon the assumption that the noise is approximately independent of the target. However, in a general scenario, at most one of $\Lambda^{\rm post}$ or $\Lambda^{\rm pre}$ will be approximately independent of the target, with the specific choice depending upon the physical implementation. As a specific example, consider amplitude damping for a single qubit, which can be written as
\begin{align}
	\Delta = \pa{\begin{array}{cccc}
		1 & 0 & 0 & 0	\\
		0 & \sqrt{g} & 0 & 0	\\
		0 & 0 & \sqrt{g} & 0	\\
		1-g & 0 & 0 & g	\\
	\end{array}}
\end{align}
in the Pauli basis $\tfrac{1}{\sqrt{2}}\pa{\unit,X,Y,Z}$, where $g\in\sq{0,1}$ determines the strength of the damping. Assume that this noise is applied independently from the left (i.e., $\Lambda^{\rm post} = \Delta$), independently of the target. Then for $X$ and $Z$,
\begin{align}
	\Lambda^{\rm pre}(Z) = \pa{\begin{array}{cccc}
	1 & 0 & 0 & 0	\\
	0 & \sqrt{g} & 0 & 0	\\
	0 & 0 & \sqrt{g} & 0	\\
	1-g & 0 & 0 & g	\\
	\end{array}}
	\quad & \quad
	\Lambda^{\rm pre}(X) = \pa{\begin{array}{cccc}
			1 & 0 & 0 & 0	\\
			0 & \sqrt{g} & 0 & 0	\\
			0 & 0 & \sqrt{g} & 0	\\
			-(1-g) & 0 & 0 & g	\\
		\end{array}}
\end{align}
which is only independent of the target when $g=1$ (i.e., when there is no noise).

In this work, we write noise operators as \emph{pre-multiplying the target} rather than post-multiplying the target as in Ref.~\cite{Magesan2011}. The reason for this change is so that the residual noise term that is not averaged is in the first time step rather than the last and so is independent of the sequence length. While this simplifies the analysis, all the results of this paper can be derived for the other form with small modifications.
			
\subsection{Measures of noise}\label{sec:fidelities}

The fidelity and trace distance between two quantum states are defined as
\begin{align}
	F(\rho,\sigma) = \bigl\lVert\sqrt\rho\sqrt\sigma\bigr\rVert_1^2 \,,\nonumber\\
	D(\rho,\sigma) = \frac{1}{2}\lVert\rho-\sigma\rVert_1\,,
\end{align}
respectively\footnote{Note that some authors define fidelity to be the square root of the fidelity defined here.}, where the 1-norm (or trace norm) is given by $\lVert X\rVert_1 = \tr\sqrt{X\ct X}$. These two quantities are related by the Fuchs-van~de~Graaf inequalities~\cite{Fuchs1999},
\begin{align}\label{eq:Fuchs}
	1-\sqrt{F(\rho,\sigma)} \le D(\rho,\sigma) \le \sqrt{1-F(\rho,\sigma)}\,,
\end{align}
where the right-hand inequality is always saturated when both states are pure. When one of the states is a pure state $\psi$, the left-hand inequality in Eq.~\eqref{eq:Fuchs} can be sharpened to
\begin{align}\label{eq:Fuchs_pure}
	1-F(\psi,\sigma) \le D(\psi,\sigma)\,.
\end{align}

Both of these quantities for quantum states can be promoted to distance measures for quantum channels~\cite{Gilchrist2005}. Two such measures are the average gate fidelity and the diamond distance.

The average gate fidelity between a channel $\E$ and a unitary $\mc{U}$ is defined to be
\begin{align}\label{eq:avg_fidelity}
	F_{\rm avg}(\E,\mc{U}) &= \int \mathrm{d}\psi F\sq{\E(\psi),\mc{U}(\psi)}	\,,
\end{align}
where $\mathrm{d}\psi$ is the unitarily invariant Haar measure. For convenience, it is typical to define a single argument version,
\begin{align}
	F_{\rm avg}(\mc{U}\ct\E) = F_{\rm avg}(\E,\mc{U})	\,,
\end{align}
which is technically the average gate fidelity between $\mc{U}\ct\E$ and $\unit$.

The \emph{diamond distance} between two quantum channels $\E_1$ and $\E_2$ with $\E_j : \D_d \to \D_d$ is defined in terms of a norm of their difference $\Delta = \E_1 - \E_2$ as follows:
\begin{align}
	\frac{1}{2} \norm{\Delta}_\diamond = \frac{1}{2} \sup_\psi \norm{\unit_d\otimes\Delta(\psi)}_1 \,.
\end{align}
The norm in the above definition is indeed a valid norm, called the diamond norm, and it extends naturally to any Hermiticity-preserving linear map between operators. The factor of $1/2$ is to ensure that the diamond distance between two channels is bounded between $0$ and $1$.

The diamond distance is useful for several reasons. Firstly, allowing for larger entangled inputs does not change the value of the diamond distance, hence it is \emph{stable}. Secondly, it has an operational meaning as determining the optimal success probability for distinguishing two unknown quantum channels $\E_1$ and $\E_2$~\cite{Kitaev1997}. Equivalently, the diamond distance gives the worst-case error rate between the pair of channels. Although we will not be able to measure the diamond distance directly, we will be able to bound it in terms of measurable quantities obtainable via randomized benchmarking.

To obtain upper and lower bounds on the diamond norm, we will use the following two lemmas to relate the average gate fidelity to the trace norm of the corresponding Choi matrix and then to the diamond norm. Recall that the Choi matrix of a linear map $\Delta$ is given by $J(\Delta) = \Delta\otimes\unit_d(\Phi)$, where $\Phi = \sum_{j,k\in\mbb{Z}_d} \ket{jj}\!\bra{kk}/d$ is the maximally entangled state. The first of these lemmas was proven in Refs.~\cite{Horodecki1999, Nielsen2002}.

\begin{lem}[\cite{Horodecki1999, Nielsen2002}]\label{lem:av_f}
	The average fidelity of a CPTP map $\E$ is related to its Choi matrix $J(\E)$ by
	\begin{align}
		(d+1)F_{\mathrm{avg}}(\E) = d F\bigl[\Phi,J(\E)\bigr] + 1 \,.
	\end{align}
\end{lem}

\begin{lem}\label{lem:diamond_bounds}
	Let $\Delta$ be a Hermiticity-preserving linear map between $d$-dimensional operators. Then the following inequalities bound the diamond norm and are saturated:
	\begin{align}
		\norm{J(\Delta)}_1 \leq \norm{\Delta}_{\diamond} \leq d\norm{J(\Delta)}_1 \,.
	\end{align}
\end{lem}

\begin{proof}
	We first prove the lower bound and show that it is saturated. We have 
	\begin{align}
		\norm{\Delta}_{\diamond} = \sup_{\psi}\norm{\Delta\otimes\unit_d(\psi)}_1
		\geq \norm{\Delta\otimes\unit_d(\Phi)}_1 = \norm{J(\Delta)}_1	\,.
	\end{align}
	To see that the above inequality is saturated, simply let $\Delta =\unit_d$.
	
	To prove the upper bound, we write
	\begin{align}
		\norm{\Delta}_{\diamond} = d \sup\Bigl\{\norm{\pa{\unit_d\otimes \sqrt{\rho_0}}J(\Delta)\pa{\unit_d\otimes \sqrt{\rho_1}} }_1:\rho_0,\rho_1\in \mc{D}_d\Bigr\}
	\end{align}
	which follows from Theorem 6 of Ref.~\cite{Watrous2012}, while being careful to note that our convention for $J(\Delta)$ differs from Ref.~\cite{Watrous2012} by a factor of $d$. Using~\cite[Prop.~IV.2.4]{Bhatia1997}, the inequality $\norm{ABC}_1 \le \norm{A}_\infty \norm{B}_1 \norm{C}_\infty$ together with $\norm{\rho}_\infty \le1$ for any state $\rho$ implies that
	\begin{align}
		\norm{\Delta}_{\diamond} &\le d \sup\Bigl\{\norm{\pa{\unit_d\otimes \sqrt{\rho_0}}}_\infty \norm{J(\Delta)}_1\norm{\pa{\unit_d\otimes \sqrt{\rho_1}}}_\infty:\rho_0,\rho_1\in \mc{D}_d\Bigr\} \le d\norm{J(\Delta)}_1	\,.
	\end{align}
	To see that this bound is saturated, let $\Delta$ be the projector onto $\density{0}$.
\end{proof}

We note that it would be interesting to see if the previous bounds are still saturated when restricting the input $\Delta$ to be a difference of channels.

\section{Time-dependent gate-independent errors in randomized benchmarking}

We consider the ideal case in which the noise depends only upon the time step. 
For such types of noise, we derive expressions for the mean $\bar{F}_m$ and 
variance $\sigma_m^2$ of the randomized benchmarking distribution 
$\br{F_{m,k}}$ for fixed $m$. In particular, we will show that for 
unital but nonunitary noise, $\sigma_m^2$ decreases exponentially with $m$, 
while for non-unital noise, $\sigma_m^2$ converges to a constant dependent on
the strength of the non-unitality. We will also upper-bound the variance for small 
$m$, which enables the derivation of rigorous confidence intervals for the 
estimate of the average gate infidelity in Sec.~\ref{sec:confidence}. We will also 
show that our results are stable under gate-dependent perturbations in the noise in Sec.~\ref{sec:perturbations}. 

In order to present our results in as clear a form as possible, we will only explicitly consider the original proposal for randomized benchmarking. Interleaved benchmarking can also be treated in an almost identical manner, except that the noise is conjugated by the interleaved gate and is redefined to absorb the noise term for the interleaved gate.

Denoting the noise at time step $t$ by $\Lambda_t$, the sequence of operations 
applied to the system in the randomized benchmarking experiment with 
sequence $s\in\mbb{N}_{\av{\G}}^m$ is
\begin{align}\label{eq:true_ini}
	\mc{S}_{s} = \prod_{t=m}^0 g_t\Lambda_t \,.
\end{align}
Here $g_t$ are the ideal unitary gates which are sampled from any unitary 2-design $\G$. 

To make it easier to analyze the above expression, we define $h_t = \prod_{b=m}^t g_b$, so that $h_0 = \unit$, $h_{m}=g_{m}$ and $g_t = h_{t+1}\ct h_t$ for all $t\in(0,m)$. Uniformly sampling the $g_t$ is equivalent to uniformly sampling the $h_t$ since $\G$ is a group; the exception is $h_0$ and $g_0$, which are chosen so that the product of all the gates is the identity (c.f.\ Sec.~\ref{sec:protocol}). We can then rewrite Eq.~\eqref{eq:true_ini} as
\begin{align}\label{eq:true_fin}
	\mc{S}_{s} = h_m\Lambda_m\ldots h_2\ct h_1\Lambda_1 h_1 \ct\Lambda_0= \prod_{t=m}^1 \Lambda_t^{h_t} \,,
\end{align}
where we incorporate the first noise term into the preparation by setting $\rho \leftarrow \Lambda_0\rho$. This redefinition of $\rho$ is independent of the sequence length because we write the noise as pre- rather than post-multiplying the target. (Note that if the noise post-multiplied the target, then incorporating the final noise term in $E$ would make $E$ depend on the sequence length $m$.)

The probability of observing the outcome $E$ for the sequence $S_{s}$ is 
$F_{m,k} = (E| \mc{S}_{s} |\rho)$. We regard the set 
$\br{F_{m,k}}$ as the realizations of a random variable with mean $\bar{F}_m$ and
variance $\sigma_m^2$. Randomized benchmarking then corresponds to randomly
sampling from the distribution $\br{F_{m,s}}$ (which we henceforth refer 
to as the randomized benchmarking distribution) to approximate the mean 
$\bar{F}_m$.

\subsection{Mean of the benchmarking distribution}\label{sec:mean}

We now derive an expression for $\bar{F}_m$ for general CPTP maps with 
time-dependent noise. A similar expression was derived for 
time-\emph{independent} noise in Ref.~\cite{Magesan2011}. We will then show 
how $\bar{F}_m$ can be used to approximate quantities of experimental 
interest, namely, the SPAM error, average time-dependent gate fidelity and the 
worst-case error due to the noise.

\begin{thm}\label{thm:fidelity_curve}
	The mean of the distribution $\br{F_{m,s}}$ for fixed $m$ is 
	\begin{align}
		\bar{F}_m =  E_0\rho_0 + \vec{E}\cdot\vec{\rho}\prod_{t=1}^{m} f_t	\,.
	\end{align}
\end{thm}

\begin{proof}
	By definition, the mean is
	\begin{align}\label{eq:time_dependent_mean}
		\bar{F}_m = \av{\mc{G}}^{-m} \sum_{s\in\mbb{N}_{\av{\G}}^m} (E|\mc{S}_{s}|\rho) =  (E| \prod_{t=m}^1 \Lambda^\G |\rho) \,.
	\end{align}
	Using Proposition \ref{prop:twirling} gives
	\begin{align}
		\bar{F}_m = \pa{\begin{array}{c c} E_0 & \vec{E} \end{array}}
		\pa{\begin{array}{c c} 1 & 0 \\ 0 & \prod_{t=1}^{m} f_t\unit_{d^2 - 1}\end{array}}
		\pa{\begin{array}{c} \rho_0 \\ \vec{\rho} \end{array}}\,.
	\end{align}
\end{proof}

The parameters $E_0\rho_0$ and $\vec{E}\vec{\rho}$ directly characterize the quality of the state and measurement procedure (with the caveat that $\rho$ has been redefined to include a noise term), since $\tr E\rho = E_0\rho_0 + \vec{E}\vec{\rho}$. This can be viewed as an instance of gate set tomography using a limited number of combinations of gates~\cite{gatesettomography}.

The parameters $f_t$ that give the mean of a randomized benchmarking distribution are closely related to an operational characterization of the amount of noise, namely, the average gate infidelity~\cite{Magesan2012, Magesan2012a} (which gives the average error rate), as
\begin{align}\label{eq:infidelity}
	r_t = 1 - F_{\mathrm{avg}}(\Lambda_t) = \frac{d-1}{d}(1-f_t)	\,.
\end{align}
The randomized benchmarking protocol will enable the estimation of $\prod_t f_t$, which can then be used to estimate the average gate infidelity averaged over arbitrary time intervals (as shown in Sec.~\ref{sec:time_dependent}). 

We now show that the average gate infidelity provides an upper and a lower bound on $\tfrac{1}{2}\norm{\Lambda-\unit}_{\diamond}$, which gives the worst-case error introduced by using $\Lambda$ instead of $\unit$. An upper bound of the same form was stated without proof in Ref.~\cite{Gambetta2012}, however, the bound here is a factor of two better. The following relation between the diamond distance and the average gate fidelity can also be applied at each time step to relate the time-averaged average gate fidelity to the average diamond distance from the identity channel. Note that the following bound is very loose in the regime $mr\ll 1$ (since in that regime, $r\ll\sqrt{r}$), which is also the regime in which we will typically use it.

\begin{prop}\label{prop:bounds}
	Let $r=1-F_{\rm avg}(\Lambda)$ be the average error rate for $\Lambda$. Then
	\begin{align}
		r(d+1)/d \leq \tfrac{1}{2}\norm{\Lambda-\unit}_{\diamond} \leq \sqrt{d(d+1)r}\,.
	\end{align}
\end{prop}

\begin{proof}
	Applying Lemma \ref{lem:diamond_bounds} to $\Delta = \Lambda-\unit$ gives
	\begin{align}
		D\sq{\Phi,J(\Lambda)}=\tfrac{1}{2}\norm{J(\Lambda) - \Phi}_1 \leq \tfrac{1}{2}\norm{\Lambda-\unit}_{\diamond} \leq \tfrac{d}{2}\norm{J(\Lambda) - \Phi}_1 = d D\sq{\Phi,J(\Lambda)}	\,.
	\end{align}
	Recalling that $\Phi$, the maximally entangled state, is a pure state and using Eq.~\eqref{eq:Fuchs} and \eqref{eq:Fuchs_pure} gives
	\begin{align}
		1- F\sq{\Phi,J(\Lambda)}	\leq D\sq{\Phi,J(\Lambda)}  \leq \sqrt{1-F\sq{\Phi,J(\Lambda)}} 	\,.
	\end{align}
	From Lemma~\ref{lem:av_f}, $1-F\sq{\Phi,J(\Lambda)} = d^{-1}(d+1)r$. Substituting this into the above expression and combining the inequalities completes the proof.
\end{proof}

\subsection{Upper bounds on the variance}\label{sec:variance}

We now consider the variance $\sigma_m^2$ of the distribution $\br{F_{m,k}}$ for fixed $m$. It has been observed that the standard error of the mean (and hence the \textit{sample} variance) can be remarkably small in experimental applications of randomized benchmarking using relatively few random sequences~\cite{Gambetta2012, Magesan2012}. In this section, we will prove that the variance due to sampling random sequences is indeed small in scenarios of practical interest (i.e., $mr\ll 1$) by obtaining an upper bound on $\sigma_m^2$ in terms of $mr$. For the special case of a qubit, we will also obtain a significantly improved upper bound in terms of $m$ and $r$. 

We begin by obtaining a general bound on $\sigma_m^2$ that depends only on $m$, $r$ and the dimension $d$ of the system being benchmarked. In order to present results in a simple form, we assume that the noise is time- and gate-independent, however, the results in this section can readily be generalized to time-dependent noise. 

As a first attempt at obtaining a good bound on the variance, we use only the fact that when $mr\ll 1$, we have $\bar{F}_m \approx A+B$, where $A = E_0\rho_0$ and $B = \vec{E}\cdot\vec{\rho}$. Expanding the expression from Theorem~\ref{thm:fidelity_curve} to first order in $r$ using Eq.~\eqref{eq:infidelity} gives
\begin{align}
	\bar{F}_m = A + B - \frac{Bmdr}{d-1}	\,.
\end{align}
The value of all realizations of $\bar{F}_m$ (i.e., the probabilities $F_{m,s}$) are all in the unit interval. Since the distribution with the largest variance that has mean $\bar{F}_m$ and takes values in the unit interval is the binomial distribution with that mean, we then have
\begin{align}
	\sigma_m^2\leq \bar{F}_m(1-\bar{F}_m) = (A + B)(1 - A - B) + \frac{mdBr}{d-1}	\,.
\end{align}
While simple to obtain, this bound has a constant off-set term that depends upon the SPAM which seems to be unavoidable. This term would be zero in the absence of SPAM, and could even be eliminated if the probabilities $F_{m,s}$ could be restricted to the interval $[1-A-B,A+B]$. However, as illustrated in Sec.~\ref{sec:confidence}, this cannot be done in general. Moreover, we expect that the above argument substantially overestimates the variance because it ignores the possibility that many sequences may have $F_{m,s}$ closer to $\bar{F}_m$.

We now obtain an alternative bound that has a larger coefficient for $r$, but no constant term. We note from the outset that the following bound is not tight in general (and the previous bound suggests that the dimensional factor is an artifact of the proof technique), though by improving one of the steps we will be able to obtain a tight bound for qubits. To facilitate our analysis, we use the identity $(E|\E |\rho)^2 = (E\nt{2}| \E\nt{2}|\rho\nt{2})$ to write the variance as
\begin{align}\label{eq:variance}
	\sigma_m^2 = \av{\G}^{-m} \sum_{k} F_{m,k}^2 - \bar{F}_m^2= (E\nt{2}|\Bigl(\sq{(\Lambda\nt{2})^\G}^m - 
	\sq{\pa{\Lambda^\G}\nt{2}}^m\Bigr)|\rho\nt{2}) \,.
\end{align}

\begin{thm}\label{thm:qudit_bound}
	The variance for time- and gate-independent randomized benchmarking of $d$-level systems with time- and gate-independent noise satisfies
	\begin{align}
		\sigma_m^2 \leq 4d(d+1)mr + O(m^2 r^2 d^4)\,.
	\end{align}
\end{thm}

\begin{proof}
	We write $\Lambda = \unit - r\Delta$, where the first row of $\Delta$ is zero since $\Lambda$ is CPTP. Since $(d^2-1)f =\tr\varphi = \tr\Lambda-1$, we can use Eq.~\eqref{eq:infidelity} to obtain $\tr\Delta = d(d+1)$
	
	We then expand the expression
	\begin{align}
		\sigma_m^2 = (E\nt{2}|\Bigl(\sq{(\Lambda\nt{2})^\G}^m - 
		\sq{\pa{\Lambda^\G}\nt{2}}^m\Bigr)|\rho\nt{2})
	\end{align}
	to second order in $r\Delta$. Note that $(\Delta\otimes\unit)^g = \Delta^g\otimes \unit$ and so all the first-order terms and the second-order terms where the $\Delta$ act at different times will cancel. Therefore the only second-order terms are the $m$ terms with $\Delta\nt{2}$ and so the variance is
	\begin{align}
		\sigma_m^2 = mr^2 (E\nt{2}|\sq{(\Delta\nt{2})^\G - 
			\pa{\Delta^\G}\nt{2}}|\rho\nt{2})	+ O(r^3\Delta^3)	\,.
	\end{align}
	Noting that $\Delta^\G = \frac{d(d+1)}{d^2-1} \unit$, the variance satisfies
	\begin{align}
		\sigma_m^2 &\leq mr^2\av{\G}^{-1}\norm{\sum_{g\in\G} (\Delta\nt{2})^g}_{\diamond} + O(r^3\Delta^3) + O(mr^2) \notag\\
		&\leq mr^2\norm{\Delta\nt{2}}_{\diamond} + O(r^3\Delta^3) + O(mr^2)	\notag\\
		&\leq mr^2\norm{\Delta}_{\diamond}^2 + O(r^3\Delta^3)	\notag\\
		&\leq 4d(d+1)mr + O(r^3\Delta^3) + O(mr^2)	\,,
	\end{align}
	where we have used the triangle inequality, the invariance of the diamond norm under unitary conjugation, the submultiplicativity of the diamond norm [with $\Delta\otimes\Delta = (\Delta\otimes\unit)(\unit\otimes\Delta)$] and Proposition~\ref{prop:bounds}.
	
	Finally, consider terms of $O(r^k\Delta^k)$ for $k>2$. For $k\geq 3$, all $O(m^k)$ such terms are upper-bounded by $r^k\norm{\Delta^k}_{\diamond}$ and so are $O(r^{k/2}d^k)$. Therefore the only contributions of $O(r^2)$ or greater are from $k=3$ and $k=4$.
	
	For $k=3$, the only terms that will not cancel are products whose only nontrivial terms are a $(\Delta\otimes\Delta)^\G$ and a $\Delta^\G\otimes\unit = \frac{d(d+1)}{d^2-1} \unit$. There are only $O(m^2)$ such terms, and applying the diamond norm bound to $\frac{d}{d-1}(\Delta\otimes\Delta)^\G$ shows that such terms contribute at most $O(m^2 r^2 d^2)$.
	
	For $k=4$, the only terms that will be of $O(r^2)$ are those that are products with two $(\Delta\otimes\Delta)^\G$ terms. Again, there are only $O(m^2)$ such terms and so such terms also contribute at most $O(m^2 r^2 d^4)$.
\end{proof}

While the bound in Theorem~\ref{thm:qudit_bound} is promising, it is not sufficiently small to justify the sequence lengths chosen in many experimental implementations of randomized benchmarking for a single qubit, since $mr\approx 10^{-2}$ in many such experiments and so the contribution to standard error of the mean due to sampling random gate sequences is expected to be on the order of $0.1 K^{-1/2}$, where $K$ is the number of random sequences of length $m$ that are sampled. 

One of the loosest approximations in Theorem~\ref{thm:qudit_bound} is the use of the triangle inequality to upper-bound the contribution from terms of the form $(\Delta\nt{2})^\G$. Avoiding this is difficult in general, however, for the case of a single qubit, we can significantly improve the following bound by understanding the irrep structure of the representation $g\otimes g$. This irrep structure will depend on the choice of 2-design, so we now fix the 2-design to be the single qubit Clifford group, $\mc{C}_2$ and work in the Pauli basis $\A=\{\unit,X,Y,Z\}/\sqrt{2}$ (where the factor of $\sqrt{2}$ makes the basis trace-orthonormal). In particular, we will work in the block basis
\begin{align}
	\pa{\begin{array}{c} \unit \\ \unit\otimes\vec{\sigma} \\ \vec{\sigma}\otimes\unit \\ \vec{\sigma}\otimes\vec{\sigma} \end{array}}
\end{align}
where $\vec{\sigma} = \{X,Y,Z\}/\sqrt{2}$. Restricting the Liouville representation to each of the blocks in the above basis will give a rep of $\mc{C_2}$, where the first three reps have already been characterized. We now characterize the final subrep, $(\phi\nt{2},\C^9)$.

\begin{prop}\label{prop:tensor_irreps}
	The representation $(\phi\nt{2},\C^9)$ of $\mc{C}_2$ is the direct sum of four inequivalent irreps.
\end{prop}

\begin{proof}
	The proof follows from a direct application of Schur's orthogonality relations, which imply
	\begin{align}\label{eq:Clifford_irreps}
		\av{\mc{C}_2}^{-1} \sum_{g\in\mc{C}_2} \chi_{\phi\nt{2}}(g)^*\chi_{\phi\nt{2}}(g) = \sum_{\lambda} n_{\lambda}^2	\,,
	\end{align}
	where $n_{\lambda}$ is the multiplicity of the irrep $\lambda$ in the rep $\phi\nt{2}$.
	
	The character is given by
	\begin{align}
		\chi_{\phi\nt{2}}(g) = \tr g\nt{2} = (\tr g)^2	\,.
	\end{align}
	Since the elements of $\mc{C}_2$ permute Paulis (up to signs), the diagonal elements of $\G$ in the Pauli basis are either $1$ or $-1$ and there are 0, 1 or 3 diagonal elements that can contribute to $\tr g$.
	
	There are eight elements of $\mc{C}_2$ with no diagonal elements, namely, the eight permutations $X\to\pm Y\to \pm Z$ and $X\to\pm Z\to \pm Z$. There is 1 element with all diagonal elements equal, namely, the identity (note that $-\unit$ is antiunitary so is not in the Clifford group). All other 15 elements of the Clifford group have $\chi_{\phi}(g) = \pm 1$ since the diagonal elements cannot sum to any values in $\br{0,\pm2, \pm3}$.
	
	Plugging these character values into Eq.~\eqref{eq:Clifford_irreps} gives
	\begin{align}
		\av{\mc{C}_2}^{-1} \sum_{g\in\mc{C}_2} \chi_{\phi\nt{2}}(g)^*\chi_{\phi\nt{2}}(g) = \frac{1}{24} \sum_{g\in\mc{C}_2} \av{\chi_{\phi}(g)}^4 = \frac{1}{24}(3^4 + 15) =4	\,.
	\end{align}
	Given that the multiplicity of an irrep must be a nonnegative integer, there are two possibilities. Either there are 4 inequivalent irreps or the rep $\phi\nt{2}$ contains two equivalent irreps. By Proposition~\ref{lem:trivial_subrep}, $\phi\nt{2}$ contains a trivial irrep with multiplicity 1 and so cannot contain two equivalent irreps.
	
	The following bases of operators:
	\begin{align}\label{eq:qubit_Schur}
		\A_1 &= \frac{1}{2\sqrt{3}}\pa{XX + YY + ZZ} \,,	\nonumber\\
		\A_2 &= \br{\frac{1}{2\sqrt{2}}\pa{XX - YY}, \frac{1}{2\sqrt{6}}\pa{XX + YY - 
				2ZZ}} \,,	\nonumber\\
		\A_S &= \frac{1}{2\sqrt{2}}\br{XY - YX, XZ - ZX, YZ - ZY} \,,	\nonumber\\
		\A_T &= \frac{1}{2\sqrt{2}}\br{XY + YX, XZ + ZX, YZ + ZY}\,,
	\end{align}
	span the four irreps.
\end{proof}

The fact that $g\otimes g$ is a direct sum of four inequivalent irreps will allow us to use Schur's lemma on the unital block of $\Delta$. To account for the nonunital component, we use the following bound.

\begin{prop}\label{prop:nonunital_qubit}
	For any completely positive and trace-preserving qubit channel $\Lambda:\D_2\to\D_2$ with average gate infidelity $r<1/3$, the nonunital part $\alpha$ obeys the inequality
	\begin{align}
		\norm{\alpha}_2^2 \leq 9r^2\,.
	\end{align}
\end{prop}

\begin{proof}
Any trace-preserving qubit channel as
\begin{align}\label{eq:ruskai_form}
	\Lambda &= (1\oplus U)\pa{\begin{array}{cccc} 
		1 & 0 & 0 & 0 \\ 
		0 & w_1 & 0 & 0 \\
		0 & 0 & w_2 & 0 \\
		t & 0 & 0 & w_3 \\
	\end{array}}(1\oplus U\ct)(1\oplus V)	\,.
\end{align}
for some $U,V\in O(3)$ [corresponding to unitaries $u,v\in U(2)$], where $\av{w_j}$ are the singular values of $\varphi$ with $\av{w_j}\in[0,1]$ for all $j$~\cite{Ruskai2002} and we have added the $(1\oplus U\ct)$ term for convenience. By Von Neumann's trace inequality~\cite{VonNeumannTraceInequality},
\begin{align}
	3 - 6r = \tr \varphi = \av{\tr UWU\ct V}  \leq \sum_j \av{w_j}
\end{align}
where $W = {\rm diag}(w_1,w_2,w_3)$ and we have used the fact that the singular values of $V$ are all one. For notational convenience, we define perturbations $\delta_j$ by $\av{w_j} = 1-\delta_j r$ which then satisfy $\sum_j \delta_j \leq 6$ and $\delta_j\geq 0$ for all $j$.

The conditions for $\Lambda$ to be completely positive are
\begin{align}
	\av{t} + \av{w_3} &\leq 1 \notag\\
	(w_j\pm w_k)^2 &\leq (1\pm w_l)^2
\end{align}
for any permutation $\br{j,k,l}$ of $\br{1,2,3}$. Therefore
\begin{align}
	\av{t} \leq 1 - \av{w_3} = \delta_3 r\,,
\end{align}
and, since $\norm{\alpha}_2$ is invariant under the unitary transformations in Eq.~\eqref{eq:ruskai_form}, $\norm{\alpha}_2 = \av{t}$. Therefore the only remaining problem is to bound $\delta_3$ (note that at this point, we could accept the trivial bound $\delta_3 \leq 6$).

If $w_3<0$, then complete positivity implies
\begin{align}
	(2-\delta_1 r - \delta_2 r)^2 \leq \delta_3^2 r^2	\,,
\end{align}
which cannot be satisfied subject to $\sum_j \delta_j \leq 6$ and $\delta_j\geq 0$ for $r<1/3$. Therefore for all $r<1/3$, $w_3 = 1-\delta_3 r>0$. 

Considering the conditions
\begin{align}
	(\delta_j - \delta_k)^2 &\leq \delta_l^2
\end{align}
for all permutations $\br{j,k,l}$ of $\br{1,2,3}$, we see that $\delta_3\leq \max_j \delta_j \leq 3$ and so $\norm{\alpha}_2 \leq 3r$.
\end{proof}

Combining the irrep structure of $g\otimes g$ and the bound on the nonunital component allows us to improve the bound in Theorem~\ref{thm:qudit_bound} for the special case of one qubit. As discussed in Sec.~\ref{sec:analysis}, the following bound provides a rigorous justification of current experiments and allows values of $K_m$ to be chosen that are substantially smaller then previously justified rigorously, that is, $K_m\approx 145$ as opposed to $K_m \approx 7\times 10^4$ as estimated in Ref.~\cite{Magesan2012a}. 

\begin{thm}\label{thm:var_qubit}
The variance for arbitrary time- and gate-independent noise satisfies
\begin{align}
	\sigma_m^2 \leq  m^2 r^2 + \frac{7}{4}mr^2 + 6\delta mr + O(m^2 r^3) + O(\delta m^2 r^2) \,,
\end{align}
where $\delta = \lvert \vec{\delta}_E \cdot\vec{\delta}_{\rho}\rvert \leq 1/2$ for any choice of
$\vec{w} \in \br{\vec{x},\vec{y},\vec{z}}$ and $\vec{\delta}_{\rho},\vec{\delta}_E\perp \vec{w}$ such that 
\begin{align}\label{eq:def:delta}
	\vec{E}^T = a\vec{w} + \vec{\delta}_E	\notag\\
	\vec{\rho} = b\vec{w} + \vec{\delta}_{\rho}	\,.
\end{align}
\end{thm}

\begin{proof}
To prove the theorem, we will derive an exact expression for the variance and then approximate it in the relevant regimes.

We begin by noting that in the basis $\br{\unit\nt{2},\unit\otimes\A,\A\otimes\unit,\A\otimes\A}$ we have
\begin{align}
	(\Lambda\nt{2})^{\mc{C}_2} = \pa{\begin{array}{cccc} 1 & 0 & 0 & 0 \\ 0 & \varphi^{\mc{C}_2} & 0 & 0\\ 
			0 & 0 & \varphi^{\mc{C}_2} & 0\\	P_1\alpha\nt{2} & \av{\mc{C}_2}^{-1}\sum_{g\in\mc{C}_2} g\alpha \otimes \varphi^{(g)}  & \av{\mc{C}_2}^{-1}\sum_{g\in \mc{C}_2} \varphi^{(g)} \otimes g\alpha & (\varphi\nt{2})^{\mc{C}_2}\end{array}}
\end{align}
where $P_1 = \av{\mc{C}_2}^{-1}\sum_{g\in \mc{C}_2} g\nt{2}$. It can be verified that $\vec{E}\nt{2}$ is in the null space of $\sum_{g\in\mc{C}_2} \varphi^{(g)} \otimes g\alpha$ and $\sum_{g\in\mc{C}_2} g\alpha \otimes \varphi^{(g)}$ for any $\vec{E}$ by, for example, considering a basis for the space of $\varphi$'s. We note in passing that this property is \textit{not} a general property of 2-designs, in that it does not hold for the single-qutrit Clifford group.

By Propositions~\ref{prop:tensor_irreps} and \ref{prop:twirling} $(\varphi\nt{2})^{\mc{C}_2} = \sum_R \lambda_R P_R$ where the $P_R$ are the projectors onto the irreps from Proposition~\ref{prop:twirling} and $\lambda_R = \tr P_R \varphi\nt{2}/\tr P_R$. From Eq.~\eqref{eq:variance}, together with the orthogonality of the projectors $P_R$, we have
\begin{align}
	\sigma_m^2 = \rho_{\unit}^2\vec{E}\nt{2}P_1\vec{\alpha}\nt{2}\sum_{t=0}^{m-1} \lambda_1^t + \sum_R \lambda_R^m \vec{E}\nt{2}P_R\vec{\rho}\nt{2} - (1-2r)^{2m}(\vec{E}\vec{\rho})^2 \,. 
\end{align}
We can bound the first term using
\begin{align}
	\rho_0^2\vec{E}\nt{2}P_1\vec{\alpha}\sum_{t=0}^{m-1} \lambda_1^t &= \frac{1}{3}\rho_0^2\norm{\vec{E}}_2^2\norm{\alpha}_2^2\sum_{t=0}^{m-1} 1 \leq \frac{3mr^2}{4}	
\end{align}
where we have used and the trivial bound $\lambda_1\leq 1$ (for the first term only) and Proposition~\ref{prop:nonunital_qubit} to obtain the final inequality.

Similarly, the eigenvalues can be calculated to be
\begin{align}
	\lambda_1 &= \frac{1}{3}\sum_{j,k} \varphi_{j,k}^2 = 
	\frac{1}{3}\tr\pa{\varphi\ct\varphi} \notag\\
	\lambda_2 &= \frac{1}{3}\sum_{j} \varphi_{j,j}^2 - 
	\frac{1}{6}\sum_{j\neq k} \varphi_{j,k}^2 = \frac{1}{2}\sum_j\varphi_{j,j}^2 - \frac{1}{6}\tr\varphi\ct\varphi 	\notag\\
	\lambda_T &= \frac{1}{6}\sum_{j\neq k} \pa{\varphi_{j,k} \varphi_{k,j} + \varphi_{j,j} \varphi_{k,k}} = \frac{1}{6}\tr\pa{\varphi^2} + \frac{1}{6}\pa{\tr\varphi}^2 - \frac{1}{3}\sum_j\varphi_{j,j}^2	\,,	
\end{align}
where we have omitted $\lambda_S$ since it will not contribute to the variance since any symmetric vector (such as $\vec{E}\nt{2}$) will be orthogonal to $P_S$. 

We now consider general noise with $\vec{E}$ and $\vec{\rho}$ as in Eq.~\eqref{eq:def:delta}, where, without loss of generality, we set $\vec{w} = \vec{z}$. We begin by considering the case $\delta = 0$, for which $\vec{E}\nt{2}P_T = \vec{E}\nt{2}P_S = 0$. Then a simple calculation using Proposition~\ref{prop:nonunital_qubit} gives $\vec{E}\nt{2}P_1 \vec{\rho}\nt{2} = a^2 b^2/3$, $(\vec{E}\vec{\rho})^2 = a^2 b^2$ and $\vec{E}\nt{2}P_2 \vec{\rho}\nt{2} = \tfrac{2}{3}a^2 b^2$.

The eigenvalues $\lambda_1$ and $\lambda_2$ can be written as $x+2y$ and $x-y$ respectively, where $x=\frac{1}{3}\sum_{j} \varphi_{j,j}^2$ and $y=\frac{1}{6}\sum_{j\neq k} \varphi_{j,k}^2$. Writing $\varphi = \unit - \Delta r$, where $\tr\Delta = 6$ (cf.\ the discussion in the proof of Theorem~\ref{thm:qudit_bound}), we have
\begin{align}\label{eq:xbound}
	1-4r+4r^2 \leq x := \frac{1}{3}\sum_{j}\varphi_{jj}^2= 1-4r + \frac{r^2}{3}\sum_j \Delta_{jj}^2 \leq 1-4r + 12r^2	\,,
\end{align}
where the maximum and the minimum are obtained by maximizing and minimizing $\sum_j \Delta_{jj}^2$ subject to $\sum_j \Delta_{jj} = 6$ for real matrices $\Delta$ with nonnegative diagonal entries respectively. The diagonal entries of $\Delta$ must be nonnegative since all entries of $\varphi$ have modulus upper-bounded by 1 [which can be easily verified from the form of extremal channels in Eq.~\eqref{eq:ruskai_form}]. Therefore the variance satisfies
\begin{align}
	\sigma_m^2 \leq \frac{3mr^2}{4} + \frac{a^2 b^2}{3}\sq{(x+2y)^m + 2(x-y)^m - 3(1-2r)^{2m}}	\,.
\end{align}
Since $1\geq \lambda_1 - 2y = x \geq 1 -4r$ by Eq.~\eqref{eq:xbound}, we have $y\leq 2r$ and so, using a binomial expansion to $O(r^3)$ gives
\begin{align}
	(x+2y)^m + 2(x-y)^m - 3(1-2r)^{2m} &\leq 12mr^2 + 12m^2r^2 + O(m^3 r^3)\,.
\end{align}
Noting that $a^2 b^2 \leq 1/4$ gives
\begin{align}
	\sigma_m^2 \leq m^2 r^2 + \frac{7}{4}mr^2 + O(m^2 r^3)\,.
\end{align}

We now consider the correction when $\delta>0$ in Eq.~\eqref{eq:def:delta}, which will realistically always be the case since $\rho$ incorporates a residual noise term. Then we define functions $h_R(\vec{\delta}_1,\vec{\delta}_2)$ by
\begin{align}
	\vec{E}\nt{2}P_R\vec{\rho}\nt{2} = a^2 b^2 \vec{z}\nt{2}P_R \vec{z}\nt{2} + h_R(\vec{\delta}_1,\vec{\delta}_2)	\,,
\end{align}
where we will henceforth omit the arguments of $h_R$. Since $\sum_R \vec{E}\nt{2}P_R\vec{\rho}\nt{2} = (\vec{E}\vec{\rho})^2$, we can write the variance as
\begin{align}
	\sigma_m^2 \leq a^2 b^2 \sigma_{m,z}^2 + \sum_R h_R \sq{\lambda_R^m - (1-2r)^{2m}}	\,.
\end{align}
To $O(r^2)$, the smallest eigenvalue is $\lambda_2$, since
\begin{align}
	\frac{1}{6}\sum_{j\neq k} \varphi_{j,j} \varphi_{k,k} &= \frac{1}{3}\sum_j \varphi_{j,j}^2 + O(r^2)	\notag\\
	\frac{1}{6}\av{\sum_{j\neq k} \varphi_{j,k} \varphi_{k,j}} &\leq \frac{1}{6}\sum_{j\neq k} \varphi_{j,k}^2
\end{align}
where the first line follows by writing $\varphi_{j,j} = 1 - r\Delta_{j,j}$ and the second from the inequality $\varphi_{j,k}^2 + \varphi_{k,j}^2 \geq 2\av{\varphi_{j,k}\varphi_{k,j}}$ and the triangle inequality. Therefore, to $O(r^2)$, $1-8r\leq x-y = \lambda_2 \leq \lambda_R \leq 1$ for all $R$ [where the bounds on $x$ and $y$ are as in Eq.~\eqref{eq:xbound}] and so $\av{\lambda_R^m - (1-2r)^{2m}}\leq 4mr + O(m^2 r^2)$ for all $R$. Therefore
\begin{align}
	\sigma_m^2 &\leq a^2 b^2 \sigma_{m,z}^2 + \sq{4mr + O(m^2 r^2)} \sum_R h_R \notag\\
	&\leq m^2 r^2 + \frac{7}{4}mr^2 + \sq{4mr + O(m^2 r^2)} \sum_R h_R + O(m^2 r^3)	 \notag\\
	&\leq m^2 r^2 + \frac{7}{4}mr^2 + 6\delta mr + O(m^2 r^3) + O(\delta m^2 r^2)
\end{align}
where we have obtained the final inequality using
\begin{align}
	a^2 b^2 + \sum_R h_R = (\vec{E}\vec{\rho})^2 = a^2 b^2 + 2ab(\vec{\delta}_1\cdot\vec{\delta}_2) + (\vec{\delta}_1\cdot\vec{\delta}_2)^2 \leq a^2 b^2 + \frac{3\delta}{2}	\,.
\end{align}
where the final inequality follows since $\delta = \lvert \vec{\delta}_1\cdot\vec{\delta}_2 \rvert, \av{ab}\leq 1/2$.
\end{proof}

It is worth noting that one could in principle fill in the implicit constants given in the big-$O$ notation by following the previous argument with sufficient care. To have a truly rigorous confidence region, one would need to take this into account, but for current parameter regimes of interest, the terms really are negligible, so it hardly seems worth optimizing this concern. 

We also note that $\delta_{\rho}$ will typically have entries of order $\sqrt{r}$ even \textit{without} SPAM, since the off-diagonal terms for generic noise are of order $\sqrt{r}$ and there is a residual noise term that has been incorporated into $\rho$. However, the corresponding entries in $\delta_E$ will generally be smaller (or at least, are determined only by SPAM).

We now show that the variance can be even further improved (by a factor of $m$ and with no dependence on the state and measurement) for noise that is diagonal in the Pauli basis.

\begin{cor}
If the unital block of the noise is diagonal in the Pauli basis, this bound can be improved to
\begin{align}
	\sigma_m^2 \leq \frac{11mr^2}{4} + O(m^2 r^3)\,.
\end{align}
\end{cor}

\begin{proof}
For noise such that $\varphi$ is diagonal in the Pauli basis, $\lambda_1 = \lambda_2 = x$ and $\lambda_T \leq \lambda_1$, which can be shown using the inequality $2ab \leq a^2 + b^2$ for $a,b\in\mbb{R}$. Therefore, for noise that is diagonal in the Pauli basis, we have
\begin{align}
	\sigma_m^2 \leq \frac{3mr^2}{4} + \frac{1}{4}\sq{\pa{1-4r + 12r^2}^m - (1-2r)^{2m}} \leq \frac{11mr^2}{4} + O(m^2 r^3)	
\end{align}
by Eq.~\eqref{eq:xbound}.
\end{proof}

One consequence of the above corollary is that the variance of the randomized benchmarking distribution will typically depend strongly upon the choice of 2-design even for gate independent noise. This observation follows from the above theorem by noting that the unital block can be perturbed by an arbitrarily small amount to allow it to be unitarily diagonalized. Performing randomized benchmarking in the basis where the unital block is diagonalized (i.e., setting $\G = \mc{C}_2^U$) will give variances of order $mr^2$, while randomized benchmarking in other bases will give variances of order $m^2 r^2$. 

\section{Asymptotic variance of randomized benchmarking}\label{sec:asymptotic}

We now consider the variance $\sigma_m^2$ of the distribution $\br{F_{m,s}}$ as $m\to \infty$. While not directly relevant to current experiments, the asymptotic behavior is nevertheless interesting in that it may provide a method of estimating the amount of nonunitality.

We will prove that, for the class of channels defined below called $n$-contractive channels (which are generic in the space of CPTP channels), $\sigma_m^2$ decays exponentially in $m$ to a constant that quantifies the amount of nonunitality. Unfortunately, we will not be able to provide a bound on the decay rate. In fact, no such bound is possible without further assumptions since the channel $[(1-\epsilon)U + \epsilon \E]^\G$ for any unitary $U$ and $2$-contractive channel $\E$ will have an eigenvalue $1-\epsilon + O(\epsilon)<1$ corresponding to the trivial subrep (this can be seen by following the proof of Proposition~\ref{lem:eigenvalues}). This eigenvalue will result in a variance that decays as $(1-\epsilon)^m$ for arbitrary $\epsilon>0$.

\begin{defn}
	A channel $\Lambda:\D_d \to \D_d$ is \emph{$n$-contractive} with respect to a group $\G\subseteq \U(d)$ if $(\Lambda\nt{n})^\G$ has at most one eigenvalue of modulus 1. 
\end{defn}

We now prove that all unital but nonunitary channels are $2$-contractive with respect to any finite 2-design. We conjecture that \textit{all} nonunitary channels are in fact $2$-contractive with respect to any unitary 2-design. An equivalent statement for trace-preserving channels $\Lambda$ is that $(\Lambda\nt{2})^{\G}$ is strongly irreducible whenever $\Lambda$ is not unitary~\cite{Sanz2009}. As a corollary of the following proposition, this conjecture holds for qubits, since, for qubits, the projection onto the unital part of a CPTP map is also a CPTP map~\cite{Kimmel2013}. However, proving it for higher dimensions remains an open problem.

\begin{prop}\label{lem:eigenvalues}
	Let $\Lambda$ be a completely positive, trace-preserving and unital channel and $\G$ a unitary 2-design. Then $\Lambda$ is $2$-contractive with respect to $\G$ if and only if it is nonunitary.
\end{prop}

\begin{proof}
First assume $\Lambda$ is unitary. Since $(\varphi,\R^{d^2-1})$ is an orthogonal 
irrep of $\U(d)$, $(\varphi,\R^{d^2-1})\nt{2}$ contains the trivial rep as a 
subrep with multiplicity 1 by Proposition~\ref{lem:trivial_subrep}. Therefore for 
any $U\in \U(d)$ and in a fixed Schur basis (i.e., independent of 
$U$), $\varphi(U)\nt{2} = 1\oplus \mc{T}(U)$ for some 
homomorphism $\mc{T}$. Therefore any vector $v$ in the 
(one-dimensional) trivial representation is a +1-eigenvector 
of $\varphi(U)\nt{2}$ for any $U$ and consequently is a 
+1-eigenvector of $\bigl[\varphi\nt{2}(\Lambda)\bigr]^\G$.

We now show that for all completely positive, trace-preserving and unital $\Lambda$,
\begin{align}
	\norm{\pa{\varphi\nt{2}}^\G}_{\infty}^2 
	&\leq 1 - \av{\G}^{-1}\sq{1 - \frac{\tr \varphi\ct\varphi}{d^2-1}}	\,.
\end{align}
Recall that one of the equivalent definitions of the spectral 
norm is
\begin{align}
	\norm{\pa{\varphi\nt{2}}^\G}_{\infty}^2 = \max_{u:\norm{u}_2=1} u\ct 
	\sq{\pa{\varphi\nt{2}}^\G}\ct \pa{\varphi\nt{2}}^\G u	\,.
\end{align}
Expanding the averages over $\G$ gives
\begin{align}
	\norm{\pa{\varphi\nt{2}}^\G}_{\infty}^2 
	&= \max_{u:\norm{u}_2=1} \av{\G}^{-2}\sum_{g,h\in\G} 
	u\ct \sq{\pa{\varphi\nt{2}}^g}\ct \pa{\varphi\nt{2}}^h u 
	\nonumber\\
	&\leq 1 - \av{\G}^{-1} + \max_{u:\norm{u}_2=1} \av{\G}^{-2}\sum_{g\in\G} 
	u\ct \sq{\pa{\varphi\ct\varphi}\nt{2}}^g u \,.
\end{align}
where in the second line we have used the improved bound for unital channels in Proposition~\ref{prop:channel_norm} to
bound the contribution from the $\av{\G}^2-\av{\G}$ terms with $g\neq h$. 

Now let $u$ be an arbitrary unit vector and write $u = \sum u_{j,k}v_j\otimes v_k$, where $\br{v_j}$ is an orthonormal basis of $\C^d$. Then, since $\pa{\varphi\ct\varphi}^g$ is positive semidefinite with 
eigenvalues upper-bounded by 1 by Proposition~\ref{prop:channel_norm}, we have
\begin{align}
	u\ct \sq{\pa{\varphi\ct\varphi}^g}\nt{2} u 
	&= \sum_{j,k} \av{u_{j,k}}^2 \sq{v_j\ct \pa{\varphi\ct\varphi}^g v_j}
	\times \sq{v_k\ct \pa{\varphi\ct\varphi}^g v_k}	\nonumber\\
	&\leq \sum_{j,k} \av{u_{j,k}}^2 v_j\ct \pa{\varphi\ct\varphi}^g v_j 
	\,,
\end{align}
where we have used $0\leq v_k\ct \pa{\varphi\ct\varphi}^g v_k \leq 1$ 
for all $k$ and $g$ to obtain the second line. By 
Proposition~\ref{prop:twirling},
\begin{align}
	\av{\G}^{-1}\sum_{g\in\G} \pa{\varphi\ct\varphi}^g  
	= \frac{\tr \varphi\ct\varphi}{d^2-1} \unit	\,,
\end{align}
so
\begin{align}
	\av{\G}^{-1}\sum_{g\in\G} u\ct \sq{\pa{\varphi\ct\varphi}^g}\nt{2} u 
	&\leq \frac{\tr \varphi\ct\varphi}{d^2-1} \sum_{j,k} \av{u_{j,k}}^2 
	\nonumber\\
	&\leq \frac{\tr \varphi\ct\varphi}{d^2-1}
\end{align}
for all $u$ such that $\norm{u}_2 = 1$.
Therefore
\begin{align}
	\norm{\pa{\varphi\nt{2}}^\G}_{\infty}^2 
	&\leq 1 - \av{\G}^{-1}\sq{1 - \frac{\tr \varphi\ct\varphi}{d^2-1}}	\,.
\end{align}
\end{proof}

\begin{thm}\label{thm:asymptotic}
	Let $\Lambda$ be a $2$-contractive channel with respect to a group $\G$ that is also a 2-design. Then the variance due to sampling random gate sequences of elements from $\G$ decays exponentially to
	\begin{align}
		\frac{\rho_0^2\vec{E}\nt{2}P_1 \alpha\nt{2}}{1-\lambda_1}	\,,
	\end{align}
	where $P_1 = \av{\G}^{-1}\sum_{g\in\G} g\nt{2}$ is a rank-1 projector and $\lambda_1 = \tr P_1 \varphi\nt{2}$.
\end{thm}

\begin{proof}
For convenience, we use the block basis 
$\br{\unit\nt{2},\unit\otimes\A,\A\otimes\unit,\A\otimes\A}$ for the matrix 
representation. In this basis, we can write
\begin{align}\label{eq:block_basis}
	(\Lambda^\G)\nt{2} &= \pa{\begin{array}{ccccc} 1 & 0 & 0 & 0 \\ 0 & f\unit 
			& 0 & 0 \\ 0 & 0 & f\unit & 0 \\
			0 & 0 & 0 & f^2 \unit \end{array}}	\notag\\
	(\Lambda\nt{2})^\G &= \pa{\begin{array}{ccccc} 1 & 0 & 0 & 0 \\ 0 & f\unit 
			& 0 & 0 \\ 0 & 0 & f\unit & 0 \\
			P_1 \alpha\nt{2} & b & c & \pa{\varphi\nt{2}}^\G \end{array}}	\,,
\end{align}
where we have used $\sum_{g\in\G} g = 0$ by Proposition~\ref{prop:twirling}, $P_1 = \av{\G}^{-1}\sum_{g\in\G} g\nt{2}$, and
\begin{align}
	b &= \av{\G}^{-1} \sum_{g\in\G} \varphi^{(g)} \otimes 
	\sq{g\alpha} \nonumber\\
	c &= \av{\G}^{-1} \sum_{g\in\G} \sq{g\alpha} \otimes \varphi^{(g)}	\,.
\end{align}
By Propositions~\ref{prop:twirling} and \ref{lem:trivial_subrep}, $P_1$ is a rank-1 projector onto the trivial subrep, which occurs with multiplicity 1. 

It can easily be shown using an inductive step that
\begin{align}
	\sq{(\Lambda\nt{2})^\G}^m = \pa{\begin{array}{ccccc} 1 & 0 & 0 & 0 \\ 0 & f^m\unit 
			& 0 & 0 \\ 0 & 0 & f^m\unit & 0 \\
			A_m & B_m & C_m & \sq{\pa{\varphi\nt{2}}^\G}^m \end{array}}	\,,
\end{align}
where
\begin{align}\label{eq:component}
	A_m &=  \sum_{t=0}^{m-1} \bgsq{(\varphi\nt{2})^\G}^t P_1 \alpha\nt{2} 	\notag\\
	B_m &= \sum_{t=0}^{m-1} f^{m - 1 -t}\sq{(\varphi\nt{2})^\G}^t b	\nonumber\\
	C_m &= \sum_{t=0}^{m-1} f^{m - 1 -t}\sq{(\varphi\nt{2})^\G}^t c  \,.
\end{align}
Since the trivial subrep occurs with multiplicity 1 and $(\varphi\nt{2})^\G$ commutes with $g\nt{2}$ for all $g$, by Proposition~\ref{prop:twirling} we can write $(\varphi\nt{2})^\G = \lambda_1 P_1 + M$ for some matrix $M$ orthogonal to $P_1$, where $\lambda_1 =\tr P_1 \pa{\varphi\nt{2}}^\G$. Therefore we have
\begin{align}
	A_m &=  P_1 \alpha\nt{2}\sum_{t=0}^{m-1} \lambda_1^t  \,.
\end{align}

Substituting these expressions into Eq.~\eqref{eq:variance} gives
\begin{align}
	\sigma_m^2 = \rho_0^2\vec{E}\nt{2}P_1 \alpha\nt{2}\sum_{t=1}^m \lambda_1^t + \rho_0\vec{E}\nt{2}(B_m+C_m)\vec{\rho} + \vec{E}\nt{2}\sq{ \pa{\varphi\nt{2}}^\G}^m \vec{\rho}\nt{2} - f^{2m}\pa{\vec{E}\vec{\rho}}^2	\,.
\end{align}

We now prove that all but the first term decay exponentially for any $2$-contractive channel with respect to $\G$. Let $SJS^{-1}$ be the Jordan decomposition of $\pa{\varphi\nt{2}}^\G$. Then, by the submultiplicativity of the spectral norm and a standard identity,
\begin{align}
	\norm{\sq{\pa{\varphi\nt{2}}^\G}^m}_{\infty}	&= \norm{\pa{SJS^{-1}}^m}_{\infty} \notag\\
	&\leq \norm{J^m}_{\infty}\norm{S}_{\infty}\norm{S^{-1}}_{\infty}	\notag\\
	&\leq (d^2-1)^2 \norm{J^m}_{\max}\norm{S}_{\infty}\norm{S^{-1}}_{\infty}
\end{align}
where $\norm{M}_{\max} = \max_{j,k}\av{M_{j,k}}$ and $J$ is a $(d^2-1)^2\times (d^2-1)^2$ matrix. Note that since $S$ is invertible, both $\norm{S}_{\infty}$ and $\norm{S^{-1}}_{\infty}$ are finite.

By explicit calculation,
\begin{align}
	\norm{J^m}_{\max} = \max_{k,j=1,\ldots,d_k} \av{\eta_k}^{m-j}\binom{m}{j}
\end{align}
where $J_k$ is the $k$th Jordan block of $J$ with eigenvalue $\eta_k$ and dimension $d_k$. By Proposition~\ref{lem:eigenvalues}, $(\Lambda\nt{2})^\G$ has at most one eigenvalue of modulus 1, which can be identified as the top left entry in the expression in Eq.~\eqref{eq:block_basis}. Consequently, all the $\eta_k$ have modulus strictly less than 1 and so $\norm{J^m}_{\max}$ and consequently $\norm{\sq{\pa{\varphi\nt{2}}^\G}^m}_{\infty}$ decay exponentially to zero with $m$.

Therefore the only term in Eq.~\eqref{eq:variance} that does not decay exponentially to zero in $m$ is the first term, namely,
\begin{align}
	\rho_0^2\vec{E}\nt{2}P_1 \alpha\nt{2}\sum_{t=1}^m \lambda_1^t	\,,
\end{align}
which converges exponentially in $m$ to
\begin{align}
	\frac{\rho_0^2\vec{E}\nt{2}P_1 \alpha\nt{2}}{1-\lambda_1}	\,,
\end{align}
provided $\av{\lambda_1}<1$, otherwise it diverges. Since $P_1 = u u\ct$ for some unit vector $u$ and the trivial rep occurs with multiplicity 1, $u$ is an eigenvector of $\pa{\varphi\nt{2}}^\G$ with eigenvalue $\lambda_1$, which must be strictly less than 1 by Proposition~\ref{lem:eigenvalues}.
\end{proof}

\section{Stability under gate-dependent 
	perturbations}\label{sec:perturbations}

In our treatment of randomized benchmarking, we have assumed that the 
noise is independent of the target gate (although the noise may depend on 
time). In a physical implementation, the noise will depend on the target. We 
can account for gate-dependent noise perturbatively by writing
\begin{align}\label{eq:perturbation}
	\Lambda_{t,g} = \Lambda_t + \epsilon\Delta_{t,g}	\,,
\end{align}
where $\Lambda_t = \av{\G}^{-1}\sum_{g\in\G} \Lambda_{t,g}$ is a valid quantum 
channel and $\epsilon$ is scaled such that $\lVert \Delta_{t,g}\rVert_{\infty} \leq 
1$ for all $t$ and $g$.

In Ref.~\cite{Magesan2012a} it was shown that the mean of the benchmarking 
distribution is robust under gate-dependent perturbations. We now show that the 
variance is also stable under gate-dependent perturbations. Let us write $\sigma_{m,0}^2$ for the gate-averaged variance, and $\delta(\sigma_m^2)$ as the correction due gate-dependent perturbations. Then we have the following theorem.

\begin{thm}\label{thm:gate-perturb}
	The gate-dependent correction to the variance satisfies $\delta(\sigma_m^2) \le \delta_0$ whenever the gate-dependent noise in Eq.~(\ref{eq:perturbation}) satisfies 
	\begin{align}
		\epsilon \le \frac{\delta_0}{9dm} \,.
	\end{align}
\end{thm}

\begin{proof}
The variance can be written in terms of an average over all gate sequences of length $m$ as
\begin{align}\label{eq:var_cor}
	\sigma_m^2 &=  -\bar{F}_m^2 + \av{\G}^{-m} \sum_{k} F_{m,k}^2	\notag\\
	&=  -\bar{F}_{m,0}^2 - \delta(\bar{F}_m^2) + \av{\G}^{-m} \sum_{k} \sq{F_{m,k,0}^2 + \delta(F_{m,k}^2)}	\notag\\
	&\leq \sigma_{m,0}^2 + 2\av{\delta(\bar{F}_m)} + 2\av{\G}^{-m} \sum_{k} \av{\delta(F_{m,k})}\notag\\
	&\leq \sigma_{m,0}^2 + 4 \av{\G}^{-m} \sum_{k} \av{\delta(F_{m,k})} \,.
\end{align}
where $M_0$ and $\delta(M)$ denote the average term and the perturbation from the average of $M$ respectively and we have used 
\begin{align}
	\av{\delta(M^2)} &= \av{[M_0 + \delta(M)]^2 - M_0^2}= \av{2M_0 + \delta(M)}\av{\delta(M)}	\leq 2\av{\delta(M)}	\,.
\end{align}
The second-to-last inequality follows since $M_0$ and $M_0 +\delta(M)$ are in the unit interval. We have also used $\av{\delta(\bar{F}_m)} \leq \av{\G}^{-1} \sum_{g\in\G}   \av{F_{m,k}}$, which follows from the triangle inequality.

With noise written in the form of Eq.~\eqref{eq:perturbation}, the sequence of operators applied in the randomized benchmarking experiment with sequence $S_{m,k}$ is
\begin{align}\label{eq:expansion}
	\mc{S}_{m,k} = \prod_{t=m}^0 g_t\Lambda_{t,g} = \sum_{a=0}^{m+1}\epsilon^a \sum_{b\in\mbb{Z}_2^{m+1}:H(b)=a}\prod_{t=m}^0 g_t M_{t,g,b_t}	\,,
\end{align}
where $H(b)$ is the Hamming weight of the bit string $b$ and
\begin{align}
	M_{t,g,b_t} = \begin{cases}
		\Lambda_t & \mbox{if } b_t = 0 \\
		\Delta_{t,g} & \mbox{if } b_t=1\,.
	\end{cases}
\end{align}
Substituting Eq.~\eqref{eq:expansion} into the expression for the
probability $F_{m,k} = (E|\mc{S}_{m,k}|\rho)$ and using the triangle inequality 
gives
\begin{align}\label{eq:corrections}
	\av{\delta(F_{m,k})} &\leq \sum_{a=1}^{m+1} \epsilon^a 
	\sum_{b}\av{(E|\pa{\prod_{t=m}^0 g_t M_{t,g,b_t}}|\rho)}	\notag\\
	&\leq \sum_{a=1}^{m+1} \epsilon^a 
	\sum_{b}\norm{\prod_{t=m}^0 g_t M_{t,g,b_t}}_{\infty}	\notag\\
	&\leq \sum_{a=1}^{m+1} \epsilon^a \binom{m+1}{a}d^{(a+1)/2}	\notag\\
	&= \sqrt{d} \bigl[ (1+\epsilon \sqrt{d})^{m+1} -1 \bigr] \notag\\
	&\leq \sqrt{d} \bigl[ \e^{\epsilon \sqrt{d} (m+1)} -1 \bigr] \,,
\end{align}
where, to get from the second to the third line, we note that for a fixed order $a$ there are at most $a+1$ quantum channels (i.e., products of $\Lambda_t$'s and the elements of $\G$, which are channels) interleaved by the $a$ different $\Delta_{t,g}$ and we have also used the submultiplicativity of the spectral norm and Proposition~\ref{prop:channel_norm}. 

We can then substitute the above sequence-independent upper bound into Eq.~\eqref{eq:var_cor}, setting $\delta_0 = \av{\delta(\sigma_m^2) - \sigma_{m,0}^2}$ and solving for $\epsilon$ gives the sufficient condition
\begin{align}
	\epsilon \le \frac{\ln(1+\delta_0/(4 \sqrt{d}))}{\sqrt{d} (m+1)} \,.
\end{align}
To extract the slightly weaker but more transparent bound stated in the theorem, we use the simple bounds $m +1 \le 2 m$ (for $m\ge 1$), $\delta_0 \le 1/4$ (because the fidelity is contained in the unit interval), the inequality $x/(1+x) \le \log(1+x)$, and the loose bound $4\sqrt{d} +\delta_0 \le 9/2 \sqrt{d}$ for $d\ge 2$.
\end{proof}

We remark that this result can surely be improved, though we have not attempted to do so. In particular, there should certainly be a factor of at least $r$, the average infidelity, bounding the change in the variance. 

\section{Conclusion}

We have proven that the randomized benchmarking protocol can be applied to experimental scenarios in which the noise is time dependent in an efficient and reliable manner. Moreover, the ability to estimate time-dependent average gate fidelities using randomized benchmarking provides an indicator for non-Markovianity over long timescales. 

In particular, we have proven that the variance is small for short sequences and asymptotically decays exponentially to a (small) constant, that, in the case of unital noise, is zero. The fact that the variance is remarkably small (e.g., on the order of $4\times 10^{-4}$ for currently achievable noise levels) enables experimental realizations of randomized benchmarking to be accurate even when using a small number of random sequences (e.g., 145 sequences compared to the $10^5$ proposed in Ref.~\cite{Magesan2012a}).

Our results show rigorously that randomized benchmarking with arbitrary Markovian noise is generically \textit{almost} as accurate as has previously been estimated in experiments~\cite{Gaebler2012, Magesan2012,Brown2011} and numerics~\cite{Epstein2014}. However, we find that $K_m$ should scale with $m$ so that the variance is independent of $m$. We also find that if near-unitary noise (such as under- and over-rotations) are a predominant noise source, then randomized benchmarking should be conducted in the regime $mr\ll 1$, since the variance due to sampling random sequences with such noise sources will remain large as $m$ increases.

It has recently been suggested that the unexpectedly good accuracy of randomized benchmarking arises because data is simultaneously fit to $\hat{F}_m$ for all sequence lengths~\cite{Epstein2014}. However, this suggestion presupposed a more fundamental fact, which we have now proven, namely, that the variance due to sampling random gate sequences is remarkably small.

Our results also apply directly to interleaved benchmarking~\cite{Magesan2012} since the interleaved gate sequence can be rewritten as a standard randomized benchmarking gate sequence with the interleaved gate and its inversion incorporated into the noise, with a consequent (but small) increase in the error rate. As such, interleaved randomized benchmarking is essentially as accurate as randomized benchmarking, provided the noise is Markovian and approximately gate-independent.

Another possible application of our results is in estimating the nonunitality of a channel by estimating the constant to which the variance asymptotically converges.  While it is not immediately apparent how to guarantee that the variance has (approximately) converged (given that the decay rate of the variance can be arbitrarily small), it may be possible to artificially boost the decay rate using, for example, the technique introduced in Ref.~\cite{Kimmel2013}. 

While our results prove that randomized benchmarking can reliably be performed using the number of sequences currently used in practice for qubits, the weaker bound on the variance for qudits implies that, to obtain results that are currently rigorously justified to a given confidence level, many more sequences are required when benchmarking higher-dimensional systems (or multiple qubits). Consequently, a major open problem is to improve the bound for qudits. Since the primary source of the improvement for qubits arose by considering the irrep structure of the tensor product representation, one route to obtaining an improved bound for qudits is to analyze the general representation structure of tensor product representations.

\textit{Acknowledgements}---JJW acknowledges helpful discussions with Joseph Emerson. The authors were supported by the IARPA MQCO program, by the ARC via EQuS project number CE11001013, and by the US Army Research Office grant numbers W911NF-14-1-0098 and W911NF-14-1-0103. STF also acknowledges support from an ARC Future Fellowship FT130101744.

\bibliography{library}

\end{document}